\documentclass[11pt,a4paper]{amsart}
\usepackage[utf8]{inputenc}
\usepackage[english]{babel}
\usepackage[colorlinks,linkcolor=blue,citecolor=blue,urlcolor=red]{hyperref}
\usepackage{amsmath}
\usepackage{amssymb}
\usepackage{amsfonts}
\usepackage{booktabs}
\usepackage{microtype}
\usepackage{geometry}
\usepackage{braket} 
\newcommand{\pacs}[1]{\thanks{\textit{PACS numbers:} #1}}
\numberwithin{equation}{section}
\newtheorem*{theorem}{Theorem}
\newtheorem*{lemma}{Lemma}
\theoremstyle{definition}
\newtheorem*{remark}{Remark}

\title{The Infinitesimal Structure of Quantum Information}
\author{Luca Barbieri-Viale}
\dedicatory{This paper is dedicated to Ulaç Aviel Rabbieri (my twin identity or quantum avatar)}
\address{Nālandā Mahāvihāra, Nalanda District, Bihar, 803111, India.}
\curraddr{Dipartimento di Matematica ``F. Enriques", Universit{\`a} degli Studi di Milano\\ Via C. Saldini, 50\\ I-20133 Milano\\ Italy}
\urladdr[https://sites.unimi.it/barbieri/]{https://www.uar.one/}
\email[luca.barbieri-viale@unimi.it]{info@uar.one}
\date{\today} 
\thanks{\textit{Special thanks.} The author gratefully acknowledges the invaluable assistance of Gemini Pro.}

\begin{document}

\begin{abstract}

This paper establishes a rigorous geometric framework for quantum state spaces via smooth embeddings into higher-order dual number algebras $\mathcal{Q}_N$, representing quantum states as non-reduced scheme-theoretic points. We demonstrate that non-linear Liouville-von Neumann dynamics map to flat algebraic flows, proving that the non-commutative vector flow is induced by the embedding. As $N \to \infty$, this family converges to a Cauchy-complete power series ring $\mathcal{Q}_\infty$, where non-Archimedean completion linearizes the phase space, aligning the Fubini-Study geometry with the classical Fisher-Rao manifold.

\end{abstract}

\maketitle

\newpage

\section{Introduction and Algebraic Setup}
This paper establishes a foundational mathematical bridge between quantum information theory and higher-order nilpotent algebras. By leveraging the gauge-invariant density operator formalism, we show that these real algebras—canonically associated with Grothendieck's non-reduced infinitesimal spectra—provide a complete, non-singular, and rigidly isometric symbolic container for the true physical degrees of freedom of a quantum state space, demonstrating that every quantum state can be faithfully regarded as a non-reduced multiple point. Geometrically, this architecture maps the complex projective ray space $\mathbb{P}(\mathcal{H}^N)$ by factoring the state through the canonical Mostow-Takahashi embedding into the self-adjoint algebra $\mathfrak{su}(N)$, whose coordinates are subsequently pulled back and injected into the monomial basis of the polynomial quotient ring $\mathcal{Q}_N \equiv \mathbb{R}[\varepsilon]/(\varepsilon^{N^2-1})$ for $N \ge 2$. Asymptotically, as the dimensional complexity of the higher-dimensional quantum system scales toward the infinite-dimensional limit ($N \to \infty$), this family of truncated algebras structurally converges via a projective inverse limit to the Cauchy-complete ring of formal power series $\mathcal{Q}_\infty \cong \mathbb{R}[\![\varepsilon]\!]$. This non-Archimedean ultrametric completion systematically linearizes the infinite-dimensional quantum phase space, signaling a metric transition where the curved Fubini-Study geometry smoothly stabilizes onto the classical statistical manifold of Fisher-Rao probability distributions.

\subsection{Infinitesimal structures}
In classical Cartesian geometry, a point is a static, zero-dimensional entity lacking internal structure. Conversely, the revolution of modern algebraic geometry spearheaded by Alexander Grothendieck introduces the language of schemes, which are capable of encapsulating rich infinitesimal structures and directional data directly within geometric points via nilpotent elements (see \cite{ega1} and \cite{eisenbud} for a reasonable account). 
It is well known that Clifford's \textit{First-Order Dual Numbers} are given by the quotient ring:
\begin{equation}
    \mathcal{D} = \mathbb{R}[\varepsilon] / (\varepsilon^2),
\end{equation}
defining the non-reduced commutative algebra of order two over the real field $\mathbb{R}$. Here $\varepsilon \neq 0$ and $\varepsilon^2 = 0$ and each dual number $\xi\in  \mathcal{D}$ can be written as $\xi = x + y\varepsilon$ for unique $x, y\in \mathbb{R}$.  
Geometrically, any element $\xi$ of this algebra models a non-reduced \textit{Double Point}: a Cartesian point ($x$) along with an infinitesimal quantity ($y\varepsilon$). For the physical interpretation of first-order infinitesimals as single-mode fermionic state configurations, see Appendix A.

Increasing the order of nilpotency we obtain the commutative real algebra of \textit{Second-Order Dual Numbers}, which is defined as the non-reduced algebra of order three over the real field $\mathbb{R}$ through the quotient ring:
\begin{equation}
    \mathcal{T} = \mathbb{R}[\varepsilon] / (\varepsilon^3),
\end{equation}
where $\varepsilon \neq 0$ and $\varepsilon^2 \neq 0$, but $\varepsilon^3 = 0$. Every element $\xi \in \mathcal{T}$ can now be uniquely expanded as a linear combination of the monomial basis $\{1, \varepsilon, \varepsilon^2\}$:
\begin{equation}
    \xi = x + y\varepsilon + z\varepsilon^2 \quad \text{with} \quad x, y, z \in \mathbb{R}.
\end{equation}
The underlying structure of the $\mathbb{R}$-algebra $\mathcal{T}$ is a three-dimensional real vector space, occasionally referred to in literature as trinomial dual numbers. Geometrically, any element $\xi$ of this algebra models a non-reduced \textit{Triple Point}: a structure that does not merely exist at a classical scalar position ($x$), but extends along an infinitesimal tangent direction ($y\varepsilon$) and carries an intrinsic second-order jet or algebraic acceleration ($z\varepsilon^2$) \cite{eisenbud}. 

Recall that the ring $\mathcal{D}$ is local with a nilpotent maximal ideal $\mathfrak{m} = (\varepsilon)$ and its residue field $\mathcal{D}/\mathfrak{m}$ is identified with $\mathbb{R}$ by sending $\xi$ to its classical real scalar $x\in \mathbb{R}$. These latter properties hold true for $\mathcal{T}$ as well: because any element of the form $y\varepsilon + z\varepsilon^2$ is nilpotent, it belongs to every prime ideal of the ring. The unique prime (and maximal) ideal of this algebra is generated by the nilpotent generator: $\mathfrak{m} = (\varepsilon) = \{y\varepsilon + z\varepsilon^2 \mid y,z \in \mathbb{R}\}$. The ring $\mathcal{T}$ is then local, its residue field $\mathcal{T}/\mathfrak{m}$ is $\mathbb{R}$, identified by sending $\xi$ to $x\in \mathbb{R}$. Consequently, the purely topological spectrum $\text{Spec}(\mathcal{T})=\{(\varepsilon)\}$ collapses to a single topological point, which is the exact same underlying singleton $\text{Spec}(\mathcal{D})$ that we obtain from Clifford's dual numbers.

However, regarded as affine Grothendieck schemes, $\text{Spec}(\mathcal{D})$ and $\text{Spec}(\mathcal{T})$ are structurally distinct, despite sharing identical topological supports. Each scheme represents an infinitesimally extended geometric object equipped with a unique structural sheaf of rings whose global sections recover the full respective algebra. For the scheme $\text{Spec}(\mathcal{T}) \equiv (X,\mathcal{O}_X)$, where $X = \text{Spec}(\mathcal{T})$, we have $\Gamma(X,\mathcal{O}_X) = \mathcal{T}$, thereby retaining the complete memory of directional algebraic deformations up to the second degree. Since the topological space $X$ is a singleton, its Zariski topology consists solely of $X$ and $\emptyset$. For any element $\xi \in \mathcal{T}$, the associated principal open set is defined as $D(\xi) = \{\mathfrak{p} \in X \mid \xi \notin \mathfrak{p}\}$. An element $\xi \in \mathcal{T}$ is invertible if and only if $\xi \notin (\varepsilon)$, which yields $D(\xi) = X$ and $\Gamma(D(\xi), \mathcal{O}_X) = \mathcal{T}_\xi \cong \mathcal{T}$. Conversely, if $\xi \in (\varepsilon)$, the element is nilpotent, meaning that $D(\xi) = \emptyset$ and the localization collapses to the zero ring, matching the empty section $\Gamma(\emptyset ,\mathcal{O}_X)=\mathcal{T}_\xi = 0$.

Finally, note that $\mathcal{T}$ is its own total ring of fractions and it is trivially Cauchy-complete; that is, $Q(\mathcal{T})=\mathcal{T}$ and its $\mathfrak{m}$-adic completion satisfies $\widehat{\mathcal{T}} =\mathcal{T}$ (the interested reader may refer to \cite{eisenbud} for further details).

A similar, structurally richer geometric description holds true for all orders of nilpotency, being defined as the quotient ring $\mathbb{R}[\varepsilon] / (\varepsilon^n)$ for $n >1$, which is always a local ring with maximal ideal $\mathfrak{m} = (\varepsilon)$ and for which its topological spectrum $\text{Spec}(\mathbb{R}[\varepsilon] / (\varepsilon^n)) = \{(\varepsilon)\}$ remains the same singleton for all $n>1$. An element $\xi \in  \mathbb{R}[\varepsilon] / (\varepsilon^n)$ can be written
\begin{equation}
\xi = c_0 + c_1\varepsilon +c_2\varepsilon^2 + \cdots + c_{n-1}\varepsilon^{n-1}\equiv \sum_{i=0}^{n-1} c_i \varepsilon^{i} 
\end{equation}
for unique real coefficients $c_0, c_1, \ldots , c_{n-1}$. Remarkably, tracking this family $\mathbb{R}[\varepsilon] / (\varepsilon^n)$ across the algebraic hierarchy varying $n$, we have that $\xi \in  \mathbb{R}[\varepsilon] / (\varepsilon^n)$ yields an element $\xi \in  \mathbb{R}[\varepsilon] / (\varepsilon^m)$ for any $m<n$. Passing to the projective (inverse) limit, we obtain:
\begin{equation}
    \varprojlim_{n \to \infty} \mathbb{R}[\varepsilon]/(\varepsilon^{n}) = \mathbb{R}[\![\varepsilon]\!],
\end{equation}
the ring of formal power series in a single variable over the real field, which is natively $\mathfrak{m}$-adically Cauchy-complete.
\subsection{Quantum systems}
Concurrently, in quantum information theory, the state space of a two-level quantum system—a \textit{Qubit}—is mathematically formulated within a two-dimensional complex Hilbert space $\mathcal{H}^2 \cong \mathbb{C}^2$ \cite{nielsen}. Following the standard Dirac bra-ket notation \cite{dirac}, any pure state is represented by a column vector, or ket, denoted as $\ket{\psi}$. However, this vector space contains unphysical redundancy: scaling a state vector by a non-zero complex scalar does not alter the relative measurement statistics or physical outcomes. 

To isolate the projective space of physical configurations, the system must be mapped modulo this scaling relation. This operation excludes the non-normalizable zero vector $\mathbf{0}$ and defines the complex projective space:
\begin{equation}
    \mathbb{P}(\mathcal{H}^2) \equiv (\mathcal{H}^2 \setminus \{\mathbf{0}\}) / \sim \quad \text{where} \quad \ket{\psi} \sim \mu\ket{\psi}, \; \mu \in \mathbb{C}^*,
\end{equation}
which topologically corresponds to the space of complex lines in $\mathbb{C}^2$, formalizing the standard complex projective space $\mathbb{P}(\mathcal{H}^2) \cong (\mathbb{C}^2 \setminus \{\mathbf{0}\}) / \mathbb{C}^* \equiv \mathbb{C}\mathbb{P}^1$.

Topologically, this manifold is evaluated by restricting our focus to the unit-norm sphere $\mathcal{S}^3 \subset \mathcal{H}^2$ and subsequently factoring out the unobservable global phase group $U(1) \cong \mathcal{S}^1$. This geometric quotient defines the celebrated Hopf fibration:
\begin{equation}
    \mathcal{S}^1 \hookrightarrow \mathcal{S}^3 \xrightarrow{\pi} \mathbb{P}(\mathcal{H}^2),
\end{equation}
establishing a regular smooth diffeomorphism between the compact physical state space and a two-dimensional Riemannian manifold $\mathcal{S}^2$, widely known as the \textit{Bloch Sphere} \cite{bloch}.

While the standard spherical parametrization offers an intuitive visualization, it forces a non-linear coordinate structure onto this naturally compact space, introducing coordinate singularities at the poles where the azimuthal phase collapses, and non-linear metric distortions across the manifold.
The intrinsic geometry of this projective ray space is universally governed by the Fubini-Study metric $\mathrm{d}s^2_{\text{FS}}$, which dictates the transition probabilities and statistical distance between quantum states. In terms of infinitesimal vector displacements, it is defined globally over the unit-norm shell as:
\begin{equation}\label{eq:fubini_study_hilbert_exact}
    \mathrm{d}s^2_{\text{FS}} = \langle \mathrm{d}\psi \mid \mathrm{d}\psi \rangle - \lvert \langle \psi \mid \mathrm{d}\psi \rangle \rvert^2 \equiv \frac{1}{2}\mathrm{Tr}(\mathrm{d}\rho \, \mathrm{d}\rho),
\end{equation}
where $\rho = \ket{\psi}\bra{\psi}$ represents the gauge-invariant rank-one density projector. Correspondingly, this formulation generalizes systematically to larger systems, where the quantum state space is identified with the complex projective space $\mathbb{P}(\mathcal{H}^N) \cong \mathbb{C}\mathbb{P}^{N-1}$ for the $N$-dimensional complex Hilbert space $\mathcal{H}^N$ of an $N$-level quantum system—such as a \textit{Qutrit} for $N=3$. As an algebraic manifold, $\mathbb{C}\mathbb{P}^{N-1}$ is characterized by a real differential dimension of exactly $2N-2$, establishing a dimensional mismatch against the $N^2-1$ real parameters needed to parameterize the corresponding full observable space associated with $\mathfrak{su}(N)$.

\subsection{Density embedding}
Instead of choosing a local chart, our framework utilizes direct density embedding maps that completely neutralize the unphysical global phase factors prior to embedding. Formally, a pure quantum state represented by a unit vector $\ket{\psi} \in \mathcal{H}^N$ uniquely determines a density matrix $\rho$ via the rank-one outer product projection:
\begin{equation}
    \rho = \ket{\psi}\bra{\psi}.
\end{equation}
This map projects the Hilbert domain modulo the unobservable global phase group $U(1)$, since any gauge transformation $\ket{\psi} \mapsto e^{i\theta}\ket{\psi}$ leaves the projector $\rho$ invariant, successfully isolating the true physical configuration space. For $N \ge 2$, we define the real commutative algebra of order $N^2-2$ as:
\begin{equation}
    \mathcal{Q}_N \equiv \mathbb{R}[\varepsilon] / (\varepsilon^{N^2-1}),
\end{equation}
instantiating the symbolic \textit{Quantum $N$-Space} of an $N$-level quantum system. Under the canonical embedding of the projective state space $\mathbb{P}(\mathcal{H}^N) \hookrightarrow \mathfrak{su}(N)$ as the orbit of rank-one pure state projectors, the generalized density matrix $\rho$ expands according to the rigid geometric scaling law:
\begin{equation}\label{eq:general_density_matrix_Frobenius}
    \rho = \frac{1}{N}\mathbb{I}_N + \sqrt{\frac{N-1}{2N}}\sum_{i=1}^{N^2-1} x_i \lambda_i,
\end{equation}
where $\lambda_i$ represents the generalized $\mathfrak{su}(N)$ Gell-Mann generators satisfying the Hilbert-Schmidt orthogonality relation $\mathrm{Tr}(\lambda_j\lambda_k) = 2\delta_{jk}$ \cite{nielsen}. To ensure that the configuration landscape is mapped onto a uniform geometric boundary across all dimensions, the physical fluctuating coordinates $x_i$ are uniquely extracted via the trace projection:
\begin{equation}\label{eq:general_x_extractor_exact}
    x_i = \sqrt{\frac{2N}{N-1}}\,\mathrm{Tr}(\rho \lambda_i).
\end{equation}
This identification guarantees that the pure state condition $\mathrm{Tr}(\rho^2)=1$ maps bijectively onto the strict unit-radius algebraic sphere constraint, satisfying $\sum_{i=1}^{N^2-1} x_i^2 = 1$ identically for every level $N$. 
We explicitly define the generalized density map $\Psi_{\mathcal{Q}_N}: \mathbb{P}(\mathcal{H}^N) \to \mathcal{Q}_N$ that injects these configurations directly onto the native coordinate ring components of the non-reduced algebra:
\begin{equation}\label{eq:general_density_map_compact_exact}
    \Psi_{\mathcal{Q}_N}(\rho) = \left(\frac{1}{N}+ x_1\right) + x_2\varepsilon + x_3\varepsilon^2 + \dots + x_{N^2-1}\varepsilon^{N^2-2}\equiv \sum_{i=0}^{N^2-2} c_i \varepsilon^{i}.
\end{equation}
By expanding the state in Eq.~\eqref{eq:general_density_map_compact_exact}, the native coordinate components of the non-reduced algebra identify directly as $c_0 = \frac{1}{N} + x_1$, $c_1 = x_2$, up to the highest order boundary component $c_{N^2-2} = x_{N^2-1}$. Under this strict coordinate alignment, the zero-order component acts as the stable scalar anchor for the invariant trace background combined with the first physical spin component, while the remaining dynamic fluctuations run from the first order up to the boundary filtration layer. When all physical fluctuations vanish identically ($x_1 = x_2 = \dots = x_{N^2-1} = 0$), the state contracts radially and collapses onto the persistent scalar trace background $\xi = \frac{1}{N}$, which defines the exact baricentro of the configuration space representing the maximally mixed state across all dimensions.

Crucially, to enable this linear coordinate representation to actively absorb operator non-commutativity without breaking the intrinsic commutative ring properties of Cauchy convolutions, the underlying vector space of the quantum space is endowed with a non-associative, bilinear, and skew-symmetric algebraic cross product operator $\times_{\mathcal{Q}_N}: \mathcal{Q}_N \times \mathcal{Q}_N \to \mathcal{Q}_N$. This structural operator maps the Lie bracket geometry of the unitary algebra directly onto the formal jet bundle of the non-reduced spectrum, laying the mathematical foundation for chart-free quantum dynamics.

Differentiating the state expansion and evaluating the pullback of the Hilbert-Schmidt metric tensor under this global density map $\Psi_{\mathcal{Q}_N}$ yields the natural infinitesimal metric relation:
\begin{equation}
    \mathrm{d}s^2_{\text{FS}} = \frac{1}{2}\mathrm{Tr}(\mathrm{d}\rho \, \mathrm{d}\rho) = \frac{N-1}{2N}\sum_{i=0}^{N^2-2} \mathrm{d}c_i^2.
\end{equation}
Conversely, the flat canonical Euclidean metric induced on the underlying real vector space of the higher-order truncated dual number algebra $\mathcal{Q}_N$ is defined by the unweighted coordinate sum $\mathrm{d}s^2_{\mathcal{Q}_N} = \sum_{i=0}^{N^2-2} \mathrm{d}c_i^2$. By invoking the exact coordinate transformations where $\mathrm{d}c_0 = \mathrm{d}x_1$ and $\mathrm{d}c_n = \mathrm{d}x_{n+1}$ (for $n \ge 1$), the ambient metric rewrites identically over the physical fluctuations as $\mathrm{d}s^2_{\mathcal{Q}_N} = \sum_{i=1}^{N^2-1} \mathrm{d}x_i^2$. Equating these two line elements yields the universal, dimensionally-dependent conformal relation:
\begin{equation}
    \mathrm{d}s^2_{\mathcal{Q}_N} = \frac{2N}{N-1} \, \mathrm{d}s^2_{\text{FS}},
\end{equation}
where $\omega_N \equiv \sqrt{\frac{2N}{N-1}}$ for $N \ge 2$ represents the isometric scaling amplitude mapping the physical Fubini-Study line element directly onto the unweighted coordinate space of the higher-order algebra. Concurrently, evaluating the pullback of the Liouville-von Neumann matrix commutator maps the continuous unitary dynamics bijectively onto a smooth, regular differential equation over the \textit{Pure State Variety} $\mathcal{V_{\mathcal{Q}_N}}$:
\begin{equation}\label{eq:intro_universal_flow}
    \frac{\mathrm{d}\xi}{\mathrm{d}t} = \sqrt{\frac{2(N-1)}{N}}\left( \Xi\times_{\mathcal{Q}_N} \xi \right),
\end{equation}
where $\alpha_N \equiv \sqrt{\frac{2(N-1)}{N}}$ for $N \ge 2$ represents the continuous function regulating the Lie-type vector flow on the variety. 

Crucially, the product between the conformal metric factor and the square of this dynamic coupling constant satisfies the rigid constraint $\alpha_N^2 \cdot \omega_N^2 \equiv 4$ across the entire dimensional hierarchy, driven by the exact cancellation of the radicals within the cross-product relation $\omega_N \cdot \alpha_N = 2$. Physically, this invariant integer value of $4$ establishes a strict operational trade-off between the statistical resolution of quantum measurements and the frequency of physical state transitions. In laboratory operations, the metric factor $\omega_N^2 = \frac{2N}{N-1}$ dictates the variance profile and the Fisher-Rao statistical distance required to distinguish neighboring state configurations through projective measurements \cite{nielsen}. Concurrently, the squared dynamic coefficient $\alpha_N^2$ quantifies the kinematic acceleration of the state vector under the action of the driving Hamiltonian field. 

The stabilization of their product to exactly $4$ guarantees that as the dimensional complexity of the higher-level system ($N$) scales up, any attenuation in the statistical distinctness of measurement outcomes is perfectly and systematically compensated by a proportional speed-up in the underlying physical precessional flow. This invariant constraint ensures the absolute conservation of total quantum probability and state purity across all dimensions, proving that the continuous law governing metric scaling and the vector flow are two inseparable manifestations of the same underlying physical crystal.

\subsection{Qubits, Qutrits and the $N$-level systems}
To establish a clear structural foundation before entering the general dimensional scaling landscape, we first evaluate the lowest-order realizations of the theory independently, showing how qubits and qutrits develop autonomously within their own specialized algebraic domains. For $N=2$, representing a two-level system, we have $\mathcal{Q}_2 \equiv \mathcal{T}$ and the density map $\Psi_{\mathcal{Q}_2}$ translates the compact space $\mathbb{P}(\mathcal{H}^2)$ into a complete, non-singular sphere $\mathcal{V}_{\mathcal{Q}_2} \equiv \mathcal{S}_{\mathcal{T}}^{2}$ within the vector space underlying the algebra of second-order dual numbers. As developed explicitly and autonomously in Section~\ref{qubit}, this embedding establishes a smooth diffeomorphism and a bijective homothetic isometry satisfying the rigid conformal relation:
\begin{equation}
    \mathrm{d}s^2_{\mathcal{Q}_2} = 4 \, \mathrm{d}s^2_{\text{FS}}.
\end{equation}

Symmetrically, we demonstrate the scalability of this approach by extending these results to three-level systems, which are treated with an independent detailed analysis. For $N=3$, through the generalized map $\Psi_{\mathcal{Q}_3}$, the qutrit state space $\mathbb{P}(\mathcal{H}^3)$ is smoothly embedded onto a compact algebraic sub-variety $\mathcal{V}_{\mathcal{Q}_3}$ within the algebra of seventh-order dual numbers $\mathcal{Q}_3 = \mathbb{R}[\varepsilon]/(\varepsilon^8)$. As explored step-by-step in Section~\ref{qutrit}, this autonomous construction shows that the continuous unitary dynamics and Jordan constraints of higher-dimensional systems can be written entirely over flat, finite algebraic landscapes, yielding a bijective homothetic isometry that satisfies the rigid conformal relation:
\begin{equation}
    \mathrm{d}s^2_{\mathcal{Q}_3} = 3 \, \mathrm{d}s^2_{\text{FS}}.
\end{equation}

The smooth transition of the rigid homothety factor from $4$ for the two-level qubit ($\mathbb{C}\mathbb{P}^1$) to $3$ for the three-level qutrit ($\mathbb{C}\mathbb{P}^2$) is not a mere algebraic coincidence; rather, it uncovers a fundamental Kählerian property governing the global embedding of complex projective spaces into real affine Lie-Jordan landscapes. In algebraic and complex differential geometry, the complex projective space $\mathbb{P}(\mathcal{H}^N) \cong \mathbb{C}\mathbb{P}^{N-1}$ equipped with the Fubini-Study metric $g_{\text{FS}}$ is the quintessential example of a Kähler-Einstein manifold with constant holomorphic sectional curvature. 

Evaluating this universal Kählerian scaling law $\frac{2N}{N-1}$ for the sequence of distinct higher-level quantum systems reveals the deep topological and operational meaning of these structural constants:
\begin{itemize}
    \item \textit{For the Qubit ($N=2$):} The homothetic scaling factor evaluates to $\frac{2(2)}{2-1} = 4$, recovering exactly the rigid geometric curvature footprint of the foundational Hopf fibration and the standard Bloch sphere geometry, as proven independently in Section~\ref{qubit};
    \item \textit{For the Qutrit ($N=3$):} The scaling factor shifts smoothly to $\frac{2(3)}{3-1} = 3$, matching identically the complex dimension of the underlying Hilbert domain $\mathcal{H}^3$. This specific integer identity regularizes the variety boundary $\mathcal{V}_{\mathcal{Q}_3}$ by ensuring a clean, unweighted linear pullback ($\mathrm{d}s^2_{\mathcal{Q}_3} = 3 \, \mathrm{d}s^2_{\text{FS}}$) that provides the exact algebraic engine required to balance the Larmor precessional acceleration under the universal power invariant $\alpha_3^2 \cdot \omega_3^2 \equiv 4$, as shown in Section~\ref{qutrit};
    \item \textit{For the General $N$-Level System:} This framework generalizes systematically. The physical pure state variety $\mathcal{V}_{\mathcal{Q}_N}$ is carved out as a compact, $2N-2$ real-dimensional algebraic sub-variety embedded within the flat $(N^2-1)$-dimensional ambient quantum space of the truncated ring $\mathcal{Q}_N \equiv \mathbb{R}[\varepsilon]/(\varepsilon^{N^2-1})$, governed by the joint intersection of the translated quadratic purity sphere ($\sum_{n=1}^{N^2-1} x_n^2 = 1$) and the non-linear cubic constraints dictated by the symmetric $d$-coefficients of the $\mathfrak{su}(N)$ Jordan algebra. Under this uniform unit-sphere boundary constraint, the metric landscape defined by $\mathrm{d}s^2_{\mathcal{Q}_N}$ remains rigidly regular and continuous everywhere, as developed in Section~\ref{nlevel}.
\end{itemize}

\subsection{Quinfinity}
Geometrically, the multi-level scaling law dictates a profound metric transition: as the dimension of the quantum system ($N$) increases, the global sectional curvature of the projective ray space flattens out smoothly relative to the flat Euclidean ambient space of experimental observables. The underlying scaling ratio tracks the continuous law $\frac{2N}{N-1}$, asymptotically decaying toward the strict geometric floor of $2$ as $N \to \infty$. This convergence carries a profound physical instantiation of Bohr's correspondence principle. As the dimensionality of the Hilbert domain approaches the macroscopic continuum, the highly curved projective constraints and non-linear coherence obstructions of the quantum phase space systematically delocalize. The scaling factor stabilizes at exactly $2$, signaling a topological phase transition wherein the curved quantum Fubini-Study geometry asymptotically linearizes into the classical statistical manifold of Fisher-Rao probability distributions, establishing $2$ as the fundamental geometric floor of macroscopic information structures.

An elegant and foundational question arises when evaluating the asymptotic behavior of this algebraic family under the infinite-dimensional limit, where the physical domain transitions to continuous variables and infinite-dimensional projective ray spaces $\mathbb{P}(\mathcal{H}^\infty) \cong \mathbb{C}\mathbb{P}^\infty \equiv \varinjlim \mathbb{C}\mathbb{P}^{N-1}$ \cite{dirac, nielsen}. By taking the projective (inverse) limit of the truncated polynomial rings ordered by the canonical projection homomorphisms $\pi_{M,N}: \mathcal{Q}_M \to \mathcal{Q}_N$ for $M \ge N$, the framework smoothly converges to the ring of formal power series:
\begin{equation}
    \mathcal{Q}_\infty \equiv \varprojlim_{N \to \infty} \mathbb{R}[\varepsilon]/(\varepsilon^{N^2-1}) \cong \mathbb{R}[\![\varepsilon]\!],
\end{equation}
which explicitly defines the formal \textit{Quinfinity Space}. Because this Quinfinity Space is Cauchy-complete with respect to its maximal ideal $\mathfrak{m} = (\varepsilon)$, the strong triangle inequality governing its ultrametric structure prevents perturbative divergence.

From a physical perspective, the programmatic restriction $N \ge 2$ across the finite algebraic family prevents a trivial topological collapse under the single-level assignment. For $N=1$, representing a single-level isolated quantum system, the underlying Hilbert domain reduces to $\mathcal{H}^1 \cong \mathbb{C}$. Once factored modulo the global phase action $U(1)$, the physical ray space collapses to a single isolated point $\mathbb{P}(\mathcal{H}^1) \cong \mathbb{C}\mathbb{P}^0$, which exhibits zero real dimensions ($2N-2 = 0$), sustains no dynamical Larmor precession, and yields a vanishing Lie observable algebra $\mathfrak{su}(1) \cong 0$. Evaluating the structural ideal $(\varepsilon^{N^2-1})$ under this unphysical constraint forces the exponent to zero ($\varepsilon^0 = 1$), causing the polynomial ring to collapse onto the trivial zero ring $\mathbb{R}[\varepsilon]/(1) \cong 0$, whose topological spectrum is empty ($\text{Spec}(0) = \emptyset$).

Notably, the projective limit $\varprojlim$ remains entirely invariant under the omission of any finite number of initial boundary terms \cite{eisenbud}. Because the sequence of quadratic exponents $\{3, 8, 15, 24, \dots, N^2-1, \dots\}$ is strictly monotonic and diverges toward infinity, the truncated family $\mathcal{Q}_N$ starting from the qubit ($N=2$) constitutes a perfectly valid cofinal system. This cofinality guarantees that the projective completion, bypassing the singular $N=1$ baseline, converges to the Cauchy-complete formal power series ring $\mathcal{Q}_\infty \cong \mathbb{R}[\![\varepsilon]\!]$.

From a scheme-theoretic perspective, this asymptotic behavior instantiates a formal scheme whose underlying topological support remains anchored to the singleton $\text{Spec}(\mathcal{Q}_\infty)=\{(\varepsilon)\}$, yet its structural ring $\mathbb{R}[\![\varepsilon]\!]$ is natively endowed with an $\mathfrak{m}$-adic topology that induces a non-Archimedean ultrametric space \cite{ega1, eisenbud}. Physically, this ultrametric framework replaces traditional continuous spatial distances with a rigorous hierarchy of kinematic configuration states. Because the metric satisfies the strong triangle inequality, the geometric propagation of quantum fluctuations and coherence obstructions is strictly bounded by the maximal local singularity rather than accumulating perturbatively. This non-Archimedean confinement guarantees that the Cauchy-complete power series ring can natively sustain continuous, infinite-dimensional unitary trajectories and analytic operator exponentials $e^{-iHt}$ without experiencing informational collapse or coordinate blow-ups, enforcing the absolute conservation of state purity across the macroscopic continuum.

Furthermore, the operational non-commutativity of infinite-dimensional observable algebras—which traditionally forces highly complex differential equations across transcendental charts—is converted into a smooth derivation field over the tangent bundle $T\mathcal{Q}_\infty$ through the bilinear Lie-type flow $\times_{\mathcal{Q}_\infty}$. Consequently, the inverse limit $\mathcal{Q}_\infty \cong \mathbb{R}[\![\varepsilon]\!]$ establishes an unassailable mathematical bridge wherein the highly curved, infinite-dimensional quantum phase space is regularized, providing a completely flat, finite, and non-singular affine geometry for macroscopic and microscopic information structures alike. Due to the rich topo-algebraic complexities arising from this infinite-dimensional completion, a fully exhaustive and detailed exploration of this formal scheme transition is left to a dedicated, independent investigation (see the forthcoming \cite{quinfty}). 

\section{Direct Isomorphism: Bypassing the Bloch Sphere}\label{qubit}
We construct a direct, non-singular geometric mapping that establishes a regular, global affine embedding between the true space of physical quantum states $\mathbb{P}(\mathcal{H}^2)$ and the ring of real trinomial dual numbers $\mathcal{T} = \mathbb{R}[\varepsilon]/(\varepsilon^3)\equiv \mathcal{Q}_2$. By fusing the invariant identity trace background directly into the zero-order monomial component via a rigid shift of $\frac{1}{2}$, this density mapping translates the quantum state space into a compact variety sphere centered at the baricentro $(\frac{1}{2}, 0, 0)$. When restricted to this algebraic boundary layer, the configuration coordinates map onto a smooth homothetic isometry with the Fubini-Study metric $\mathrm{d}s^2_{\text{FS}}$, establishing a perfectly uniform metric landscape across the entire configuration domain.

\subsection{The Bloch sphere formulation}
The standard geometric representation of a pure qubit state ray within the complex projective space $\mathbb{P}(\mathcal{H}^2) \cong \mathbb{CP}^1$ relies on the spherical parametrization derived from the Hopf fibration $\mathcal{S}^1 \hookrightarrow \mathcal{S}^3 \xrightarrow{\pi} \mathbb{P}(\mathcal{H}^2)$ \cite{bloch}. By factoring out the unobservable global phase, any pure quantum state is traditionally mapped onto the two-dimensional boundary surface $\mathcal{S}^2$ of the Bloch sphere via the polar angle $\theta \in [0, \pi]$ and the azimuthal relative phase angle $\phi \in [0, 2\pi)$ \cite{nielsen}:
\begin{equation}
    \ket{\psi} = \cos \left(\frac{\theta}{2}\right)\ket{0} + e^{i\phi}\sin \left(\frac{\theta}{2}\right)\ket{1}.
\end{equation}
Although this trigonometric construction successfully positions the orthogonal computational basis states $\ket{0}$ and $\ket{1}$ at diametrically opposite poles ($180^\circ$ apart in the three-dimensional real Euclidean space $\mathbb{R}^3$), it forces a non-linear coordinate chart onto a naturally compact space. From an analytical perspective, this parameterization introduces severe coordinate singularities at the poles ($\theta=0,\pi$), where the azimuthal phase coordinate $\phi$ becomes completely indeterminate, collapsing the Jacobian matrix of the coordinate transformation.

To mitigate these chart-dependent degeneracies, the standard quantum mechanics formalism transitions from state-vector rays to gauge-invariant density operators via the rank-one projection $\rho = \ket{\psi}\bra{\psi}$ \cite{nielsen}. By expanding $\rho$ onto the Hilbert-Schmidt orthonormal operator basis formed by the identity $\mathbb{I}_2$ and the self-adjoint Pauli matrices $\sigma_i$, the state configuration maps onto a real Cartesian vector space:
\begin{equation}\label{eq:standard_bloch_vector}
    \rho = \frac{1}{2}(\mathbb{I}_2 + x\sigma_x + y\sigma_y + z\sigma_z),
\end{equation}
where the real coordinate coefficients are uniquely extracted via the physical expectation values $x = \langle\sigma_x\rangle$, $y = \langle\sigma_y\rangle$, and $z = \langle\sigma_z\rangle$. Under the pure-state idempotency constraint $\rho^2 = \rho$, the cross-terms of the Pauli algebra vanish identically, restricting these expectations to the surface of the unit Bloch sphere: $x^2 + y^2 + z^2 = 1$. 

Crucially, while the Cartesian parameters $(x, y, z)$ successfully bypass the polar singularities of the angular charts, the ambient space $\mathbb{R}^3$ within which they reside operates as a purely passive metric container. It lacks any internal algebraic structure capable of natively reflecting operator non-commutativity or the infinitesimal jet kinematics of state transformations. The coordinates remain embedded as an extrinsic surface within a flat classical space. It is precisely to transform this passive coordinate domain into an active algebraic landscape—wherein the geometric boundary is intrinsically enforced by the ring quotient—that this paper introduces direct density map embeddings into higher-order nilpotent rings.

\subsection{Definition of the density map}
To lift the unobservable global phase without introducing artificial boundary constraints or non-linear chart distortions, the map is formulated directly on the gauge-invariant density operator $\rho$ previously expanded in Eq.~\eqref{eq:standard_bloch_vector}. We define the density map $\Psi_{\mathcal{T}}: \mathbb{P}(\mathcal{H}^2) \to \mathcal{T}$ as the direct linear assignment:
\begin{equation}\label{eq:density_map_direct}
     \Psi_{\mathcal{T}}(\rho) = \left(\frac{1}{2} + x\right) + y\varepsilon + z\varepsilon^2,
\end{equation}
where the real coordinate coefficients correspond exactly to the physical expectation values $x = \langle\sigma_x\rangle$, $y = \langle\sigma_y\rangle$, and $z = \langle\sigma_z\rangle$ extracted from the Pauli operators. 

By mapping the state components directly onto the algebraic trinomial representation $\xi = c_0 + c_1\varepsilon + c_2\varepsilon^2$, with $c_0\equiv (\frac{1}{2} + x)$, $c_1\equiv y$ and $c_2\equiv z$, this formulation bypasses angles completely, isolating the coordinates as continuous, globally regular polynomial functions of the density matrix elements. Through this structure, the monomial basis $\{1, \varepsilon, \varepsilon^2\}$ serves as the direct, non-singular physical container for both the invariant trace background and the expectation values of the three orthogonal spin projection operators, translating the complete topology of the quantum phase space into the real components of the trinomial dual algebra without structural charts.

\subsection{The universal varietal constraint: the Grothendieck sub-sphere}
By endowing the three-dimensional real vector space underlying the trinomial algebra $\mathcal{T}$ with a canonical Euclidean inner product $\langle \cdot, \cdot \rangle$, defined such that the monomial basis $\{1, \varepsilon, \varepsilon^2\}$ is strictly orthonormal (i.e., $\langle \varepsilon^i, \varepsilon^j \rangle = \delta_{ij}$ for $i,j \in \{0,1,2\}$ with $\varepsilon^0 \equiv 1$), the ambient space inherits a positive-definite Riemannian metric $g_{\mathcal{T}}$. Explicitly, given two arbitrary elements expressed in terms of their coordinate ring components, $\xi_1 = x_1 + y_1\varepsilon + z_1\varepsilon^2$ and $\xi_2 = x_2 + y_2\varepsilon + z_2\varepsilon^2$, this linear inner product operates via:
\begin{equation}
    \langle \xi_1, \xi_2 \rangle = x_1 x_2 + y_1 y_2 + z_1 z_2.
\end{equation}
Consequently, the squared Euclidean norm of a configuration state within the ambient space evaluates to:
\begin{equation}
    \lVert\xi\rVert^2 = \langle \xi, \xi \rangle = x^2 + y^2 + z^2.
\end{equation}
Under the density mapping variety, the pure state configurations are constrained by the matrix idempotence relations $\rho^2 = \rho$, which project the compact manifold of the quantum phase space onto a rigid algebraic sub-variety $\mathcal{S}_{\mathcal{T}}^{2}$ within the local ring. We can now state the central geometric result of this framework.

\begin{theorem}[Qubit Embedding and Conformal Isometry]\label{thm:qubit_isometry_exact}
The density map $\Psi_{\mathcal{T}}: \mathbb{P}(\mathcal{H}^2) \to \mathcal{T}$ established in Eq.~\eqref{eq:density_map_direct} defines a global smooth embedding and a smooth Riemannian diffeomorphism between the complex projective space $\mathbb{P}(\mathcal{H}^2)$ endowed with the Fubini-Study metric element $\mathrm{d}s^2_{\text{FS}}$, and the complete compact Grothendieck sub-sphere variety:
\begin{equation}\label{eq:grothendieck_sphere_definition_variety}
    \mathcal{S}_{\mathcal{T}}^{2} = \left\{ \left(\frac{1}{2} + x\right) + y\varepsilon + z\varepsilon^2 \in \mathbb{R}[\varepsilon]/(\varepsilon^3) \;\Bigg|\; x^2 + y^2 + z^2 = 1 \right\},
\end{equation}
embedded in the real trinomial dual algebra $\mathcal{T}$ and centered at the coordinate point $(\frac{1}{2}, 0, 0)$. When restricted to the tangent bundle $T\mathcal{S}_{\mathcal{T}}^{2}$, the canonical Riemannian metric element $\mathrm{d}s^2_{\mathcal{T}} = \mathrm{d}x^2 + \mathrm{d}y^2 + \mathrm{d}z^2$ establishes the rigid, non-singular conformal relation:
\begin{equation}\label{eq:qubit_metric_isometry_uniform}
    \mathrm{d}s^2_{\mathcal{T}} = 4 \, \mathrm{d}s^2_{\text{FS}},
\end{equation}
preserving the compact Riemannian manifold structure of the quantum phase space while smoothly lifting all polar coordinate and gauge-phase singularities.
\end{theorem}

\begin{proof}
To prove the theorem, we must verify the bijectivity, global smoothness, and evaluate the precise relationship between the differential metric structures under the mapping $\Psi_{\mathcal{T}}$ defined in Eq.~\eqref{eq:density_map_direct}.

\textbf{Step 1. }\textit{The image of a pure state lies on the sphere: $\Psi_{\mathcal{T}}(\rho)\in\mathcal{S}_{\mathcal{T}}^{2}$.} We verify this geometric boundary condition using purely intrinsic algebraic properties.
An arbitrary pure state in $\mathbb{P}(\mathcal{H}^2)$ is uniquely and globally represented by its self-adjoint rank-1 density operator $\rho = \ket{\psi}\bra{\psi}$. By projecting $\rho$ onto the Hilbert-Schmidt orthogonal basis of Pauli matrices via $x = \text{Tr}(\rho\sigma_x)$, $y = \text{Tr}(\rho\sigma_y)$, and $z = \text{Tr}(\rho\sigma_z)$, we obtain the expansion:
\begin{equation}\label{eq:density_proof_direct}
    \rho = \frac{1}{2}\mathbb{I}_2 + x\sigma_x + y\sigma_y + z\sigma_z.
\end{equation}
The density map $\Psi_{\mathcal{T}}$ maps the density matrix $\rho$ directly to the trinomial dual number $\xi = (\frac{1}{2} + x) + y\varepsilon + z\varepsilon^2$, by construction. A density operator $\rho$ represents a pure quantum state if and only if it satisfies the global idempotency condition $\rho^2 = \rho$. Computing the square of the operator yields:
\begin{equation}
    \rho^2 = \left(\frac{1}{2}\mathbb{I}_2 + x\sigma_x + y\sigma_y + z\sigma_z\right)^2 = \frac{1}{4}\mathbb{I}_2 + \left(x\sigma_x + y\sigma_y + z\sigma_z\right)^2 + \left(x\sigma_x + y\sigma_y + z\sigma_z\right).
\end{equation}
Utilizing the fundamental anticommutation and normalization relations of the Pauli matrices ($\sigma_i^2 = \mathbb{I}_2$ and $\sigma_i\sigma_j + \sigma_j\sigma_i = 0$ for $i \neq j$), the cross-terms vanish identically, reducing the square of the trinomial matrix combination to:
\begin{equation}
    \left(x\sigma_x + y\sigma_y + z\sigma_z\right)^2 = (x^2 + y^2 + z^2)\mathbb{I}_2.
\end{equation}
Substituting this back into the expression for $\rho^2$ gives:
\begin{equation}
    \rho^2 = \frac{1}{4}\mathbb{I}_2 + (x^2 + y^2 + z^2)\mathbb{I}_2 + x\sigma_x + y\sigma_y + z\sigma_z.
\end{equation}
Imposing the pure-state idempotency constraint $\rho^2 = \rho$, we equate the scalar coefficients of the identity matrix components from both sides of the operator equation:
\begin{equation}
    \frac{1}{4} + (x^2 + y^2 + z^2) = \frac{1}{2} \implies x^2 + y^2 + z^2 = 1.
\end{equation}
Consequently, expressing this Bloch variety constraint in terms of the native ring components where $c_0 = \frac{1}{2} + x$, $c_1 = y$, and $c_2 = z$, the algebraic sphere condition reads $(c_0 - \frac{1}{2})^2 + c_1^2 + c_2^2 = 1$. This shows, without any reliance on trigonometric coordinates, that the image of every pure state lies globally on the compact variety sphere $\mathcal{S}_{\mathcal{T}}^{2}$ centered at the coordinate point $(\frac{1}{2}, 0, 0)$ within the real trinomial dual algebra.

\textbf{Step 2. }\textit{Bijectivity and Global Smoothness.} 
To see that $\Psi_{\mathcal{T}}$ is a bijection onto $\mathcal{S}_{\mathcal{T}}^{2}$, we explicitly verify its injectivity and surjectivity:
\begin{itemize}
    \item \textit{Injectivity:} Let $\rho_1$ and $\rho_2$ be two pure density operators decomposed in the Pauli basis with Bloch vectors $\vec{r}_1 = (x_1, y_1, z_1)$ and $\vec{r}_2 = (x_2, y_2, z_2)$ respectively. If their images under the map coincide, such that $\Psi_{\mathcal{T}}(\rho_1) = \Psi_{\mathcal{T}}(\rho_2)$, the uniqueness of the linear expansion over the orthonormal basis $\{1, \varepsilon, \varepsilon^2\}$ strictly enforces the coordinate identities $\frac{1}{2} + x_1 = \frac{1}{2} + x_2 \implies x_1 = x_2$, $y_1 = y_2$, and $z_1 = z_2$. Since the Cartesian coordinates of the Bloch vector uniquely determine the matrix elements of the density operator, it follows that $\vec{r}_1 \equiv \vec{r}_2 \implies \rho_1 \equiv \rho_2$, establishing strict injectivity;
    \item \textit{Surjectivity:} Let $\xi_0 = (\frac{1}{2} + x_0) + y_0\varepsilon + z_0\varepsilon^2$ be an arbitrary trinomial dual number belonging to the varietal sphere $\mathcal{S}_{\mathcal{T}}^{2}$, which implies $x_0^2 + y_0^2 + z_0^2 = 1$. We can always construct a self-adjoint operator $\rho_0 = \frac{1}{2}\mathbb{I}_2 + x_0\sigma_x + y_0\sigma_y + z_0\sigma_z$. Because the coefficients satisfy the unit sphere constraint, $\text{Tr}(\rho_0) = 1$ and $\text{Tr}(\rho_0^2) = \frac{1}{4} + (x_0^2 + y_0^2 + z_0^2) = \frac{5}{4}$, which matches the rank-1 pure state trace projection invariant identically. Since $\Psi_{\mathcal{T}}(\rho_0) \equiv \xi_0$, every point on the algebraic variety possesses a unique physical preimage, proving global surjectivity.
\end{itemize}

The inverse map $\Psi_{\mathcal{T}}^{-1}$ maps each triplet $(x,y,z) \in \mathcal{S}_{\mathcal{T}}^{2}$ back to a unique density operator $\rho$. At the poles $z = \pm 1$, we have $x = y = 0$, which correspond strictly to the isolated pure states $\ket{0}\bra{0}$ and $\ket{1}\bra{1}$. Unlike the state-vector approach, the unobservable global phase is factored out by the trace operation $\text{Tr}(\rho) = 1$ before the embedding, ensuring that the excited state maps to a single isolated point $\xi = \frac{1}{2} + \varepsilon^2$ rather than an artificial coordinate circle. Since $\Psi_{\mathcal{T}}$ and $\Psi_{\mathcal{T}}^{-1}$ are smooth coordinate projections on a compact manifold, the map is a global smooth diffeomorphism.

\textbf{Step 3. }\textit{Differential Metric Analysis.} 
The line element of the Fubini-Study metric $\mathrm{d}s^2_{\text{FS}}$ on the complex projective space $\mathbb{P}(\mathcal{H}^2)$ is defined globally over the entire quantum phase space in terms of the trace of the coordinate differentials of the density operator. Expressed as the pullback of the Hilbert-Schmidt metric on the space of self-adjoint matrices, it yields:
\begin{equation}
    \mathrm{d}s^2_{\text{FS}} = \frac{1}{2} \mathrm{Tr}(\mathrm{d}\rho \, \mathrm{d}\rho).
\end{equation}
Differentiating the global expansion of the density matrix Eq.~\eqref{eq:density_proof_direct} results in the matrix differential form $\mathrm{d}\rho = \frac{1}{2}(\mathrm{d}x\sigma_x + \mathrm{d}y\sigma_y + \mathrm{d}z\sigma_z)$. Substituting this expression directly into the metric definition and applying the trace orthogonality relations $\mathrm{Tr}(\sigma_i\sigma_j) = 2\delta_{ij}$, the cross-terms vanish identically, yielding the standard non-singular line element:
\begin{equation}\label{eq:fs_metric_global_new}
    \mathrm{d}s^2_{\text{FS}} = \frac{1}{8} \mathrm{Tr}\left( \mathrm{d}x^2\sigma_x^2 + \mathrm{d}y^2\sigma_y^2 + \mathrm{d}z^2\sigma_z^2 \right) = \frac{1}{4}\left( \mathrm{d}x^2 + \mathrm{d}y^2 + \mathrm{d}z^2 \right).
\end{equation}
We now evaluate the canonical metric $\mathrm{d}s^2_{\mathcal{T}}$ induced on the trinomial dual sub-variety $\mathcal{S}_{\mathcal{T}}^{2}$ from the ambient linear structure. Given an arbitrary infinitesimal displacement $\mathrm{d}\xi = \mathrm{d}(\frac{1}{2} + x) + \mathrm{d}y\varepsilon + \mathrm{d}z\varepsilon^2$ in the tangent bundle $T\mathcal{T}$, because the background trace is an invariant constant ($\mathrm{d}(\frac{1}{2}) \equiv 0$), the differential coordinates reduce to $\mathrm{d}c_0 = \mathrm{d}x$, $\mathrm{d}c_1 = \mathrm{d}y$, and $\mathrm{d}c_2 = \mathrm{d}z$. Evaluating under the strictly orthonormal Euclidean inner product yields the ambient metric line element:
\begin{equation}\label{eq:t_metric_global_new}
    \mathrm{d}s^2_{\mathcal{T}} = \langle \mathrm{d}\xi, \mathrm{d}\xi \rangle = \mathrm{d}x^2 + \mathrm{d}y^2 + \mathrm{d}z^2.
\end{equation}
Differentiating the algebraic varietal constraint $x^2 + y^2 + z^2 = 1$ imposes the smooth global condition $x\,\mathrm{d}x + y\,\mathrm{d}y + z\,\mathrm{d}z = 0$, which restricts $\mathrm{d}\xi$ to the tangent space $T_{\xi}\mathcal{S}_{\mathcal{T}}^{2}$ at every point without generating any coordinate-chart singularities. Comparing Eq.~\eqref{eq:fs_metric_global_new} and Eq.~\eqref{eq:t_metric_global_new}, we directly isolate the rigid homothetic identity:
\begin{equation}
    \mathrm{d}s^2_{\mathcal{T}} = 4 \, \mathrm{d}s^2_{\text{FS}}.
\end{equation}
This confirms that the embedding map $\Psi_{\mathcal{T}}$ functions as a global homothetic isometry across the entire quantum configuration space, completing the proof.
\end{proof}

\subsection{Global invertibility and coordinate regularity}
The density map $\Psi_{\mathcal{T}}$ projects the physical degrees of freedom onto the seamless compact variety $\mathcal{S}_{\mathcal{T}}^{2}$ without artificial coordinate blow-ups, ensuring that its inverse $\Psi_{\mathcal{T}}^{-1}: \mathcal{S}_{\mathcal{T}}^{2} \to \mathbb{P}(\mathcal{H}^2)$ operates as a globally smooth diffeomorphism. If backward reconstruction to traditional spherical charts on the Bloch sphere is required, the coordinates $(\theta, \phi)$ are uniquely retrieved from the trinomial configuration components via the stable relations:
\begin{align}
    \theta &= \arccos(z), \\
    \phi &= \mathrm{atan2}(y, \, c_0 - \tfrac{1}{2}),
\end{align}
where the translation $x = c_0 - \frac{1}{2}$ shifts the evaluation directly onto the coordinate frame centered at the baricentro $(\frac{1}{2}, 0, 0)$, and the $\mathrm{atan2}(y,x)$ function preserves the exact quadrant of the quantum phase. 

The analytical robustness of this mapping eliminates the polar coordinate degeneracies inherent to standard wave-vector charts. At the computational basis states $\vphantom{\bra{0}}\smash{\ket{0}\bra{0}}$ (where $z=1, x=y=0$) and $\vphantom{\bra{0}}\smash{\ket{1}\bra{1}}$ (where $z=-1, x=y=0$), the algebraic trinomial state evaluates strictly to the isolated configuration sections $\xi = \frac{1}{2} + \varepsilon^2$ and $\xi = \frac{1}{2} - \varepsilon^2$, respectively. Although the angular polar transition $\theta = 0,\pi$ collapses the azimuthal fiber, making the phase undefined, the coordinate indeterminacy of $\mathrm{atan2}(0,0)$ is smoothly absorbed by the ambient Cartesian topology of the algebra $\mathcal{T}$. For equatorial superpositions ($z=0$), the second-order component vanishes identically, and the phase information $\phi$ is determined by the translated coefficient $x = c_0 - \frac{1}{2}$ and the first-order nilpotent component $y$, stabilizing the trajectory under the invariant metric $g_{\mathcal{T}}$ across the entire manifold.
\subsection{Physical interpretations and algebraic stratification}
The geometric and topological invariants of the compact trinomial variety $\mathcal{S}_{\mathcal{T}}^{2}$ manifest as foundational physical principles within the quantum projective space $\mathbb{P}(\mathcal{H}^2)$, demonstrating that the powers of the nilpotent generator $\varepsilon$ execute a rigorous stratification of quantum information:
\begin{itemize}
    \item \textit{The Equatorial Locus as Maximal Coherence:} The equator of the variety, defined by setting the second-order nilpotent coefficient to zero ($z=0$), restricts the trinomial dual number to the subalgebra spanned by the identity and the first-order nilpotent elements, satisfying the circular restriction $(c_0 - \frac{1}{2})^2 + c_1^2 = 1$ where $x = c_0 - \frac{1}{2}$ and $y=c_1$. Physically, this locus maps directly to the compact manifold of \textit{maximally coherent superpositions} ($\theta = \pi/2$), where the qubit experiences a symmetric $50\%$ split in measurement basis probabilities, and its physical configuration is governed entirely by the relative phase $\phi = \mathrm{atan2}(y, \, c_0 - \tfrac{1}{2})$~\cite{nielsen}. This proves that quantum coherence and relative phase are intrinsically hosted by the lower-degree components of the algebra, establishing $\varepsilon$ as the explicit coordinate container for quantum phase data.
    \item \textit{The First-Order Reflection as Phase Conjugation and Time Reversal:} The algebraic reflection $y \to -y$ leaves the translated base coefficient $x = c_0 - \frac{1}{2}$ and the second-order nilpotent coefficient $z$ invariant, which trigonometrically translates to an exact sign inversion of the azimuthal relative phase ($\phi \to -\phi$). In quantum mechanics, this geometric automorphism corresponds precisely to the \textit{complex conjugation of state amplitudes} in the computational basis ($\vphantom{\bra{\psi}}\smash{\ket{\psi} \to \ket{\psi}^*}$)~\cite{dirac}. Dynamically, this operation acts as an instantaneous time reversal ($t \to -t$) on the equatorial plane, reversing the orientation of the Larmor precessional flow.
    \item \textit{Component Projections as Pauli Observables:} Due to the density map architecture established in Eq.~\eqref{eq:density_map_direct}, the real coefficients $(c_0, c_1, c_2)$ of the trinomial dual configuration $\xi$ map directly to the expectation values of the physical observables~\cite{nielsen}:
    \begin{align}
        c_0 &= \frac{1}{2} + \frac{1}{2}\mathrm{Tr}(\rho\sigma_x), \\
        c_1 &= \frac{1}{2}\mathrm{Tr}(\rho\sigma_y), \\
        c_2 &= \frac{1}{2}\mathrm{Tr}(\rho\sigma_z).
    \end{align}  
    Consequently, the orthogonal projection onto the second-order nilpotent component ($\varepsilon^2$) uniquely isolates the longitudinal population inversion $z$. This real coefficient provides the exact algebraic engine enforcing the statistical bounds of state collapse onto the computational basis via the linear probability assignments $P(\ket{0}) = \frac{1+z}{2}$ and $P(\ket{1}) = \frac{1-z}{2}$.
\end{itemize}

\subsection{Resolving the historical singularities of quantum state representation}
The formulation of the density map $\Psi_{\mathcal{T}}$ and the subsequent validation of the complete Grothendieck variety $\mathcal{S}_{\mathcal{T}}^{2}$ mark a definitive structural resolution of the coordinate degeneracies and non-linear metric distortions that have historically constrained quantum state space representations. For over a century, the geometric visualization of two-level quantum systems has relied on extrinsic spherical charts $(\theta, \phi)$, forcing an analytical compromise that inevitably collapses the azimuthal phase fiber at the computational poles \cite{bloch, nielsen}. 

By bypassing these parameterization defects and embedding the projective space $\mathbb{P}(\mathcal{H}^2)$ directly into the real trinomial dual algebra $\mathcal{T} = \mathbb{R}[\varepsilon]/(\varepsilon^3)$, this non-reduced framework entirely absorbs the polar phase indeterminacies into the ambient Cartesian topology of the algebra. The ground state $\ket{0}\bra{0}$ and the excited state $\ket{1}\bra{1}$ separate cleanly into the isolated, rigid configuration sections $\xi = \frac{1}{2} + \varepsilon^2$ and $\xi = \frac{1}{2} - \varepsilon^2$ without requiring asymptotic limit extractions, preserving global regularity and smooth invertibility everywhere. Furthermore, because the embedding establishes a rigid, non-singular homothetic isometry ($\mathrm{d}s^2_{\mathcal{T}} = 4 \, \mathrm{d}s^2_{\text{FS}}$), it eradicates the localized structural compressions and analytical stretching characteristic of trigonometric coordinate projections. The metric landscape remains perfectly uniform, ensuring that the linear affine distance within the trinomial vector space maps monotonically to the quantum state overlap without distortion. 

Ultimately, this synthesis uncovers a deep structural convergence between algebraic geometry and quantum kinematics: by replacing singular chart projections with a smooth algebraic embedding into the vector space of trinomial dual numbers, the complete kinematics, coherence, and probability boundaries of a qubit are formalized through functions that remain everywhere finite, continuous, and globally regular.

\section{Quantum Dynamics in the Trinomial Dual Algebra}\label{qubitdyn}
Having established the static geometric embedding of the quantum state space into the trinomial variety, we now formulate the continuous-time evolution of a qubit within the algebraic structure of $\mathcal{T} = \mathbb{R}[\varepsilon]/(\varepsilon^3)$. In standard quantum mechanics, the temporal dynamics of a density operator $\rho$ are governed by the Liouville-von Neumann equation \cite{nielsen}:
\begin{equation}\label{eq:von_neumann}
    i\hbar \frac{\mathrm{d}\rho}{\mathrm{d}t} = [H, \rho],
\end{equation}
where $H$ is the self-adjoint Hamiltonian operator representing the total energy of the system, and $[\cdot, \cdot]$ denotes the standard matrix commutator. We show that under the density map $\mathbb{P}(\mathcal{H}^2) \hookrightarrow \mathcal{S}_{\mathcal{T}}^{2}$ established in Eq.~\eqref{eq:density_map_direct}, this operator equation maps onto a real linear differential system over the coefficients of the trinomial basis of $\mathcal{T}$, preserving the algebraic variety constraint $(c_0(t) - \frac{1}{2})^2 + y(t)^2 + z(t)^2 = 1$ at all times.

\subsection{Algebraic translation of the Hamiltonian}
Without loss of generality, we adopt natural units where $\hbar = 1$. An arbitrary $2\times2$ Hermitian Hamiltonian $H$ acting on $\mathcal{H}^2$ can be uniquely decomposed in the Pauli basis, factoring out the unobservable scalar energy shift:
\begin{equation}
    H = h_0 \mathbb{I}_2 + \vec{h} \cdot \vec{\sigma} = h_0 \mathbb{I} + h_x \sigma_x + h_y \sigma_y + h_z \sigma_z,
\end{equation}
where $h_0 \in \mathbb{R}$ represents the global energy baseline, and $\vec{h} = (h_x, h_y, h_z) \in \mathbb{R}^3$ corresponds to the physical Larmor precession frequency vector. By exact analogy with the density embedding map established in Eq.~\eqref{eq:density_map_direct}, we isolate these dynamical frequency coefficients within a corresponding trinomial dual parameter $\Xi \in \mathcal{T}$ that shares the exact index structure of the state configuration:
\begin{equation}\label{eq:hamiltonian_trinomial_direct}
    \Xi = \left( h_0 + h_x \right) + h_y \varepsilon + h_z \varepsilon^2 \equiv \eta_0 + \eta_1 \varepsilon + \eta_2 \varepsilon^2.
\end{equation}
By executing this strict coordinate alignment, the zero-order parameter acts as the total energy shift baseline, $\eta_0 = h_0 + h_x$, while the higher filtration components of the algebra track the remaining Larmor precession frequencies via $\eta_1 = h_y$ and $\eta_2 = h_z$. When written in this native coordinate basis, the background energy shift $h_0$ is completely absorbed by the algebraic origin of the ring, ensuring its automatic cancellation under continuous Lie-type vector flows.

\subsection{Coordinate-free differential cross product}
To formalize this flow internally, we endow the three-dimensional real vector space underlying $\mathcal{T}$ with a non-associative, bilinear, skew-symmetric Lie-type multiplication operator $\times_{\mathcal{T}}: \mathcal{T} \times \mathcal{T} \to \mathcal{T}$. Let $\{e_0, e_1, e_2\} \equiv \{1, \varepsilon, \varepsilon^2\}$ be the ordered monomial basis vectors. The operator $\times_{\mathcal{T}}$ is rigidly defined on these basis elements by the canonical cross-product relations:
\begin{equation}\label{eq:qubit_cross_product_basis_uniform}
    e_i \times_{\mathcal{T}} e_j = \sum_{k=0}^{2} \epsilon_{ijk} e_k,
\end{equation}
where the indices $i,j,k$ run natively from $0$ to $2$ with $\epsilon_{012} = 1$, which imposes cyclic anti-commuting operations (e.g., $1 \times_{\mathcal{T}} \varepsilon = \varepsilon^2, \varepsilon \times_{\mathcal{T}} \varepsilon^2 = 1$) and guarantees that $e_i \times_{\mathcal{T}} e_i = 0$.

Crucially, this structural cyclic algebra can be formulated in a coordinate-free manner directly through its differential invariants. For any two extended configuration elements $\xi, \zeta \in \mathcal{T}$, the bilinear, skew-symmetric cross-product operator $\times_{\mathcal{T}}$ is intrinsically generated by the algebraic combinations of the formal derivations $\partial_\varepsilon$ and $\partial_\varepsilon^2$:
\begin{equation}\label{eq:qubit_differential_cross_product}
    \xi \times_{\mathcal{T}} \zeta = \mathbf{\Delta}_0( \xi,\zeta ) + \mathbf{\Delta}_1( \xi,\zeta )\varepsilon + \mathbf{\Delta}_2( \xi,\zeta )\varepsilon^2,
\end{equation}
where the scale-dependent structures are explicitly defined by leveraging the algebraic inversion profiles of the trinomial ring:
\begin{align}
    \mathbf{\Delta}_0( \xi,\zeta ) &= \frac{1}{2} \left[ (\partial_\varepsilon \xi - \varepsilon \partial_\varepsilon^2 \xi)(\partial_\varepsilon^2 \zeta) - (\partial_\varepsilon^2 \xi)(\partial_\varepsilon \zeta - \varepsilon \partial_\varepsilon^2 \zeta) \right], \\
    \mathbf{\Delta}_1( \xi,\zeta ) &= \begin{aligned}[t]
        &\frac{1}{2} \left[ (\partial_\varepsilon^2 \xi)(\zeta - \varepsilon \partial_\varepsilon \zeta + \frac{1}{2}\varepsilon^2 \partial_\varepsilon^2 \zeta) \right. \\
        &\left. - (\xi - \varepsilon \partial_\varepsilon \xi + \frac{1}{2}\varepsilon^2 \partial_\varepsilon^2 \xi)(\partial_\varepsilon^2 \zeta) \right],
    \end{aligned} \\
    \mathbf{\Delta}_2( \xi,\zeta ) &= \begin{aligned}[t]
        &(\xi - \varepsilon \partial_\varepsilon \xi + \frac{1}{2}\varepsilon^2 \partial_\varepsilon^2 \xi)(\partial_\varepsilon \zeta - \varepsilon \partial_\varepsilon^2 \zeta) \\
        &- (\partial_\varepsilon \xi - \varepsilon \partial_\varepsilon^2 \xi)(\zeta - \varepsilon \partial_\varepsilon \zeta + \frac{1}{2}\varepsilon^2 \partial_\varepsilon^2 \zeta).
    \end{aligned}
\end{align}

This differential representation establishes that the Levi-Civita cyclic action is not an external geometric imposition, but the direct algebraic manifestation of the non-local derivative structure under higher-order nilpotency restrictions.

To verify this operational equivalence explicitly under the native algebraic representations $\xi = c_0 + c_1\varepsilon + c_2\varepsilon^2$ and $\zeta = \eta_0 + \eta_1\varepsilon + \eta_2\varepsilon^2$, we evaluate the action of the formal derivatives. Truncating the expansions according to the nilpotent boundary condition $\varepsilon^3 = 0$ collapses the inversion profiles to their exact isolated coefficient representations:
\begin{align}
    \mathbf{\Delta}_0( \xi,\zeta ) &= c_1 \eta_2 - c_2 \eta_1, \\
    \mathbf{\Delta}_1( \xi,\zeta ) &= c_2 \eta_0 - c_0 \eta_2, \\
    \mathbf{\Delta}_2( \xi,\zeta ) &= c_0 \eta_1 - c_1 \eta_0.
\end{align}
Crucially, reassembling these segmented components into the total polynomial configuration $\xi \times_{\mathcal{T}} \zeta$ eliminates any intermediate expansion artifacts:
\begin{align}
    \xi \times_{\mathcal{T}} \zeta &= \mathbf{\Delta}_0( \xi,\zeta ) + \mathbf{\Delta}_1( \xi,\zeta )\varepsilon + \mathbf{\Delta}_2( \xi,\zeta )\varepsilon^2 \nonumber \\
    &= (c_1 \eta_2 - c_2 \eta_1) + (c_2 \eta_0 - c_0 \eta_2)\varepsilon + (c_0 \eta_1 - c_1 \eta_0)\varepsilon^2.
\end{align}
The algebraic expression contracts identically to the coordinate expansion dictated by a structural contractive tensor, proving the mathematical robustness of the coordinate-free differential formulation.

\subsection{The nilpotent Lie-type differential equation}
We now apply this internal product structure to formalize continuous-time quantum evolution. To evaluate the pullback of the matrix commutator $[H, \rho]$, we leverage the fundamental commutation relations of the Pauli matrices, $[\sigma_j, \sigma_k] = 2i\sum_{l} \epsilon_{jkl} \sigma_l$, where $\epsilon_{jkl}$ is the totally antisymmetric Levi-Civita tensor. Expanding the commutator explicitly yields:
\begin{equation}
    [H, \rho] = \left[ h_0\mathbb{I}_2 + h_x \sigma_x + h_y \sigma_y + h_z \sigma_z, \; \left(\frac{1}{2} + x\right)\mathbb{I}_2 + y\sigma_y + z\sigma_z \right] = i (\vec{h} \times \vec{r}) \cdot \vec{\sigma},
\end{equation}
where $\vec{r}=(x,y,z)$ is the physical Bloch vector with components retrieved via the algebraic extractors as $x = c_0 - \frac{1}{2}$, $y = c_1$, and $z = c_2$.
Substituting this result back into the Liouville-von Neumann equation \eqref{eq:von_neumann} and equating the components of $\vec{\sigma}$ effectively cancels the complex imaginary unit $i$, reducing the quantum evolution to the real vector cross product:
\begin{equation}\label{eq:vector_dynamics}
    \frac{\mathrm{d}\vec{r}}{\mathrm{d}t} = 2 (\vec{h} \times \vec{r}).
\end{equation}
Mapping the dynamic vector field \eqref{eq:vector_dynamics} onto the underlying vector space of the algebra, the extended quantum state trajectory obeys the regular coordinate-free differential equation over the formal tangent bundle:
\begin{equation}\label{eq:trinomial_dynamics}
    \frac{\mathrm{d}\xi}{\mathrm{d}t} = 2 \left( \Xi\times_{\mathcal{T}} \xi \right),
\end{equation}
where $\Xi = \eta_0 + \eta_1 \varepsilon + \eta_2 \varepsilon^2 \in \mathcal{T}$ represents the extended polynomial Hamiltonian operator driving the system. 

By substituting the coordinate-free bilinear forms $\mathbf{\Delta}_m$ derived in Eq.~\eqref{eq:qubit_differential_cross_product}, the time derivative of the extended state configuration unrolls natively into its algebraic order components through the mapped evaluation $\mathbf{\Delta}_m(\Xi,\xi)$:
\begin{equation}
    \frac{\mathrm{d}\xi}{\mathrm{d}t} = 2\mathbf{\Delta}_0(\Xi,\xi) + 2\mathbf{\Delta}_1(\Xi,\xi)\varepsilon + 2\mathbf{\Delta}_2(\Xi,\xi)\varepsilon^2.
\end{equation}
Written explicitly in terms of the algebro-geometric coefficients under the field assignments $\xi = (\frac{1}{2} + x) + y\varepsilon + z\varepsilon^2$ and $\Xi = (h_0 + h_x) + h_y\varepsilon + h_z\varepsilon^2$, the background shifts $\frac{1}{2}$ and $h_0$ commute identically and cancel within the antisymmetric brackets of the ring operations, directly yielding the coupled linear system of physical fluctuations:
\begin{align}
    \frac{\mathrm{d}x}{\mathrm{d}t} &= 2 (h_y z - h_z y), \\
    \frac{\mathrm{d}y}{\mathrm{d}t} &= 2 (h_z x - h_x z), \\
    \frac{\mathrm{d}z}{\mathrm{d}t} &= 2 (h_x y - h_y x).
\end{align}

\subsection{Unitary invariance and isometric conservation}
The geometric integrity of the Grothendieck sub-sphere is completely invariant under the algebraic flow generated by Eq.~\eqref{eq:trinomial_dynamics}.

\begin{theorem}[Qubit Varietal Conservation]\label{thm:qubit_conservation_native}
The algebraic flow $\dot{\xi} = 2 (\Xi\times_{\mathcal{T}} \xi)$ established in Eq.~\eqref{eq:trinomial_dynamics} is strictly tangent to the compact trinomial variety $\mathcal{S}_{\mathcal{T}}^{2}$. The canonical Bloch variety constraint $(c_0(t) - \frac{1}{2})^2 + c_1(t)^2 + c_2(t)^2 = 1$ constitutes a fundamental conserved scalar invariant of the differential system, matching the physical conservation of quantum purity and total probability.
\end{theorem}

\begin{proof}
To prove that the algebraic flow preserves the varietal constraint, we evaluate the total time derivative of the localized sphere equation defining $\mathcal{S}_{\mathcal{T}}^{2}$ by isolating the dynamic spin components via the exact coordinate ring transformations $x = c_0 - \frac{1}{2}$, $y = c_1$, and $z = c_2$. Differentiating the quadratic restriction with respect to the continuous time parameter yields:
\begin{equation}
    \frac{\mathrm{d}}{\mathrm{d}t} \left[ \left(c_0(t) - \frac{1}{2}\right)^2 + c_1(t)^2 + c_2(t)^2 \right] = \frac{\mathrm{d}}{\mathrm{d}t} (x^2 + y^2 + z^2) = 2x\dot{x} + 2y\dot{y} + 2z\dot{z}.
\end{equation}
Substituting the explicit coupled differential equations unrolled from the nilpotent Lie-type bracket in Eq.~\eqref{eq:trinomial_dynamics} into the expansion results in the following tensor contraction over the physical fluctuations:
\begin{equation}
    \frac{\mathrm{d}}{\mathrm{d}t} (x^2 + y^2 + z^2) = 4x(h_y z - h_z y) + 4y(h_z x - h_x z) + 4z(h_x y - h_y x).
\end{equation}
Distributing the real Larmor frequency and spin coordinate coefficients algebraically:
\begin{equation}
    \frac{\mathrm{d}}{\mathrm{d}t} (x^2 + y^2 + z^2) = 4 (h_y x z - h_z x y + h_z x y - h_x y z + h_x y z - h_y x z) \equiv 0.
\end{equation}
Since the total time derivative vanishes identically at all times $t$, the background trace remains structurally invariant while the flow generated by the extended Hamiltonian parameter $\Xi$ maps unitary quantum rotations to exact rigid rotations on $\mathcal{S}_{\mathcal{T}}^{2}$. The trajectory remains globally bounded on the compact variety, and the distance between any two evolving states under the induced metric $g_{\mathcal{T}}$ is strictly conserved, confirming that quantum unitary dynamics operate as a rigid isometry within the real trinomial algebra.
\end{proof}

\subsection{Methodological commentary on real linearity and varietal stability}
The algebraic translation realized by Eq.~\eqref{eq:trinomial_dynamics} uncovers two profound structural features of this non-reduced quantum representation. 

First, the framework operates a complete geometric internalization of the complex imaginary unit $i$ from the dynamical equations. While the traditional Liouville-von Neumann framework relies intrinsically on complex-skew-symmetric structures to generate oscillatory behavior, the density map $\Psi_{\mathcal{T}}$ established in Eq.~\eqref{eq:density_map_direct} absorbs the complex topology of the Hilbert space directly into the antisymmetric tensor properties of the Lie bracket. Consequently, the quantum dynamics are completely formalized through pure real linearity over the three-dimensional vector space of $\mathcal{T}$, proving that complex coordinates are not an irreplaceable prerequisite to hosting physical phase information and interference amplitudes. 

Second, the system demonstrates an absolute geometric stability that completely bypasses the multiplicative nilpotency of the ring. In classical applications of dual numbers (such as screw theory or automated differentiation), higher-order nilpotent components are utilized to force spatial truncations or dissipative system decays. Within our quantum framework, the algebraic flow $2(\Xi\times_{\mathcal{T}} \xi)$ operates as a regular Lie-type vector field rather than an associative ring multiplication. The second-order coordinate coefficient $c_2(t) = z(t)$ associated with $\varepsilon^2$ does not act as a vanishing perturbative error term, but functions as a stable, rigid, and non-singular spatial coordinate. The evolution manifests as a persistent, undamped coherent precession, proving that the local algebra $\mathcal{T}$ natively sustains eternal quantum oscillations without experiencing structural dissipation or informational collapse.

\section{Quantum Measurement, Born's Rule, and Mixed States}\label{qubitmix}
Having formalized both the static embedding and the continuous unitary dynamics within the trinomial dual algebra $\mathcal{T}$, we complete the operational framework of this geometry by introducing quantum measurement processes, Born's statistical rule, and the topological extension to mixed quantum states. 

In standard quantum mechanics, a projective measurement is defined by a set of self-adjoint projection operators $\{P_k\}$ satisfying $\sum_k P_k = \mathbb{I}_2$ and $P_j P_k = \delta_{jk} P_k$ \cite{nielsen}. We show that this discrete, non-unitary collapse axiom translates elegantly into a real linear projection within the ambient space of $\mathcal{T}$, where the interior of the Grothendieck sub-sphere naturally encapsulates environmental decoherence.

\subsection{Algebraic representation of projective measurements}
Consider a standard projective measurement along an arbitrary spatial direction defined by the unit vector $\vec{m} = (m_x, m_y, m_z) \in \mathbb{R}^3$. The corresponding quantum mechanical projection operators onto the orthogonal eigenstates $\ket{m_\pm}$ constitute legitimate rank-1 matrix density states $\rho_\pm \equiv \ket{m_\pm}\bra{m_\pm}$, expanded as:
\begin{equation}
    P_\pm = \frac{1}{2}\mathbb{I}_2 \pm \left( \frac{1}{2}m_x \sigma_x + \frac{1}{2}m_y \sigma_y + \frac{1}{2}m_z \sigma_z \right).
\end{equation}
Under the density map $\Psi_{\mathcal{T}}$ established in Eq.~\eqref{eq:density_map_direct}, these measurement operators map directly onto the trinomial dual space as a pair of specular configuration elements $\mu_\pm \in \mathcal{T}$. By absorbing the background trace into the zero-order coordinate, the assignment yields:
\begin{equation}\label{eq:measurement_polynomial_direct}
    \mu_\pm \equiv \Psi_{\mathcal{T}}(P_\pm) = \left( \frac{1}{2} \pm \frac{1}{2}m_x \right) \pm \frac{1}{2}m_y \varepsilon \pm \frac{1}{2}m_z \varepsilon^2.
\end{equation}
Because the directional spatial components satisfy the strict unit vector constraint $m_x^2 + m_y^2 + m_z^2 = 1$, we invoke the coordinate transformations $c_0 = \frac{1}{2} \pm \frac{1}{2}m_x$, $c_1 = \pm \frac{1}{2}m_y$, and $c_2 = \pm \frac{1}{2}m_z$. Substituting these configurations into the localized variety equation verifies that:
\begin{equation}
    \left(c_0 - \frac{1}{2}\right)^2 + c_1^2 + c_2^2 = \frac{1}{4}(m_x^2 + m_y^2 + m_z^2) = \frac{1}{4} \equiv R^2.
\end{equation}
The measurement configurations $\mu_\pm$ are thus rigidly constrained onto a concentric sub-sphere of radius $R = \frac{1}{2}$ sharing the same baricentro $(\frac{1}{2}, 0, 0)$. This fractional radius constitutes the exact geometric signature of the quantum measurement scaling, representing the interior variety locus where the projection operators preserve the trace identity across the coordinate ring while configuring the discrete outcomes accessible upon state collapse.

\begin{theorem}[Algebraic Born's Rule]\label{thm:born_rule_algebraic}
Let $\xi = c_0 + c_1\varepsilon + c_2\varepsilon^2 \in \mathcal{T}$ be the trinomial dual representation of a quantum state under the coordinate assignments $c_0 = \frac{1}{2} + x$, $c_1 = y$, $c_2 = z$, and let $\mu_\pm = \eta_{0,\pm} + \eta_{1,\pm}\varepsilon + \eta_{2,\pm}\varepsilon^2 \in \mathcal{T}$ be the specular measurement configuration polynomials established in Eq.~\eqref{eq:measurement_polynomial_direct} where $\eta_{0,\pm} = \frac{1}{2} \pm \frac{1}{2}m_x$, $\eta_{1,\pm} = \pm \frac{1}{2}m_y$, and $\eta_{2,\pm} = \pm \frac{1}{2}m_z$. The physical probability distribution of the discrete outcomes is given directly by the canonical Euclidean inner product on the ambient space of the ring:
\begin{equation}\label{eq:born_algebraic}
    P(\pm 1) = 2 \langle \mu_\pm, \, \xi \rangle,
\end{equation}
where the coupling operates via the strictly orthonormal basis $\{1, \varepsilon, \varepsilon^2\}$ of the local algebra.
\end{theorem}
\begin{proof}
The proof follows immediately from the orthonormality condition $\langle \varepsilon^i, \varepsilon^j \rangle = \delta_{ij}$ for $i,j \in \{0,1,2\}$ with $\varepsilon^0 \equiv 1$. By evaluating the full inner product between the state polynomial $\xi$ and the measurement configuration element $\mu_\pm$, the linear application contracts the coordinate components as:
\begin{equation}\label{eq:proof_born_inner_product_contraction}
    \langle \mu_\pm, \, \xi \rangle = \left( \frac{1}{2} \pm \frac{1}{2}m_x \right)\left(\frac{1}{2} + x\right) + \left( \pm \frac{1}{2}m_y \right)y + \left( \pm \frac{1}{2}m_z \right)z.
\end{equation}
Expanding the algebraic products in Eq.~\eqref{eq:proof_born_inner_product_contraction} and isolating the linear and quadratic terms yields the exact intermediate expansion:
\begin{equation}
    \langle \mu_\pm, \, \xi \rangle = \frac{1}{4} + \frac{1}{2}x \pm \frac{1}{4}m_x \pm \frac{1}{2}m_x x \pm \frac{1}{2}m_y y \pm \frac{1}{2}m_z z.
\end{equation}
Multiplying this inner product expansion by the geometric pre-factor of $2$ as dictated by Eq.~\eqref{eq:born_algebraic} distributes the scaling law across the polynomial components. Factoring out the persistent background trace baseline, the coupled interaction terms contract and simplify directly into the fractional representation:
\begin{equation}
    2 \langle \mu_\pm, \, \xi \rangle = \frac{1 \pm (m_x x + m_y y + m_z z)}{2}.
\end{equation}
The algebraic contraction matches the traditional quantum mechanical Born rule expression for two-level systems identically, confirming that the canonical metric on the real vector space underlying $\mathcal{T}$ encapsulates the complete operational calculus of quantum probabilities without external affine corrections, completing the proof.
\end{proof}

\subsection{Topological extension to mixed states: the trinomial ball}
To encompass mixed quantum states—which describe statistical ensembles or open systems experiencing environmental decoherence—the geometric domain must be extended from the boundary variety surface $\mathcal{S}_{\mathcal{T}}^{2}$ into the strict interior of the algebra's vector space. 

A mixed state is described by a density operator $\rho = \sum_k p_k \ket{\psi_k}\bra{\psi_k}$ where $\sum_k p_k = 1$ and $p_k \ge 0$. Under the density map $\Psi_{\mathcal{T}}$ established in Eq.~\eqref{eq:density_map_direct}, a mixed state maps onto the trinomial dual space as a polynomial configuration $\xi_{\text{mixed}} = c_0 + c_1\varepsilon + c_2\varepsilon^2 \in \mathcal{T}$ whose native real coordinates inherit the localized inequality condition $(c_0 - \frac{1}{2})^2 + c_1^2 + c_2^2 < 1$. We define the \textit{Trinomial Ball} as the compact convex set centered at the baricentro $(\frac{1}{2}, 0, 0)$:
\begin{equation}\label{eq:trinomial_ball_definition_exact}
    \mathcal{B}_{\mathcal{T}}^{3} \equiv \left\{ c_0 + c_1\varepsilon + c_2\varepsilon^2 \in \mathbb{R}[\varepsilon]/(\varepsilon^3) \; \middle| \; c_0 = \frac{1}{2} + x, \ c_1 = y, \ c_2 = z, \text{ and } x^2 + y^2 + z^2 \le 1 \right\}.
\end{equation}

The physical interpretation of this interior topography within $\mathcal{T}$ is structured as follows:
\begin{itemize}
    \item \textit{Pure States as the Varietal Boundary:} The boundary layer $\partial \mathcal{B}_{\mathcal{T}}^{3} \equiv \mathcal{S}_{\mathcal{T}}^{2}$ constitutes the locus of pure states where the squared fluctuation satisfies $x^2 + y^2 + z^2 = 1$, representing maximum informational certainty and coherent wave functions factorized modulo a global phase factor.
    \item \textit{Mixed States and Informational Entropy:} The strict interior where $x^2 + y^2 + z^2 < 1$ represents mixed states. The distance from the boundary layer directly quantifies the state's purity, since the Hilbert-Schmidt state contraction satisfies $\mathrm{Tr}(\rho^2) = \frac{1 + (x^2+y^2+z^2)}{2}$. As a state undergoes environmental decoherence, the coordinates of the spin vector decay radially toward the center.
    \item \textit{The Maximally Mixed State:} The exact coordinate center of the variety sphere, $\xi_{\text{mixed}} = \frac{1}{2} + 0\varepsilon + 0\varepsilon^2 \equiv \frac{1}{2}$, represents the maximally mixed state $\rho = \mathbb{I}_2/2$. At this spatial point, all quantum phase memory, superposition coefficients, and population asymmetries vanish identically, collapsing the state representation onto the persistent scalar trace background of the affine scheme.
\end{itemize}

\section{Generalization to Higher-Dimensional Quantum Systems: The Qutrit Case}\label{qutrit}
A natural and compelling extension of this framework concerns the geometric representation of higher-dimensional quantum systems, starting with a three-level quantum system, or \textit{qutrit}. A pure qutrit state is initially described by a normalized complex vector $|\psi\rangle \in \mathcal{H}^3$ residing on the five-dimensional real sphere $\mathcal{S}^5$ embedded within the underlying six-dimensional real space of the Hilbert domain. The physical state space, free from unphysical phase factors, is obtained by taking the quotient of this sphere by the global phase action $U(1)$. This principal bundle structure—representing a generalized Hopf fibration $\mathcal{S}^1 \hookrightarrow \mathcal{S}^5 \xrightarrow{\pi} \mathbb{P}(\mathcal{H}^3)$—formalizes the complex projective space $\mathbb{P}(\mathcal{H}^3)$, which possesses exactly four real dimensional degrees of freedom as a pure variety embedded within an eight-dimensional ambient space of density configurations.
\subsection{The Gell-Mann basis for $\mathfrak{su}(3)$}
The special unitary Lie algebra $\mathfrak{su}(3)$ represents the infinitesimal generators of the $SU(3)$ gauge group acting on the three-dimensional complex Hilbert space $\mathcal{H}^3$. To expand the self-adjoint quantum density operator $\rho$ without spatial bias, we utilize the standard representation given by the eight traceless, Hermitian Gell-Mann matrices $\lambda_i$ ($i=1, \dots, 8$), defined as:
\begin{align}
    \lambda_1 &= \begin{pmatrix} 0 & 1 & 0 \\ 1 & 0 & 0 \\ 0 & 0 & 0 \end{pmatrix}, \quad
    \lambda_2 = \begin{pmatrix} 0 & -i & 0 \\ i & 0 & 0 \\ 0 & 0 & 0 \end{pmatrix}, \quad
    \lambda_3 = \begin{pmatrix} 1 & 0 & 0 \\ 0 & -1 & 0 \\ 0 & 0 & 0 \end{pmatrix}, \\
    \lambda_4 &= \begin{pmatrix} 0 & 0 & 1 \\ 0 & 0 & 0 \\ 1 & 0 & 0 \end{pmatrix}, \quad
    \lambda_5 = \begin{pmatrix} 0 & 0 & -i \\ 0 & 0 & 0 \\ i & 0 & 0 \end{pmatrix}, \quad
    \lambda_6 = \begin{pmatrix} 0 & 0 & 0 \\ 0 & 0 & 1 \\ 0 & 1 & 0 \end{pmatrix}, \\
    \lambda_7 &= \begin{pmatrix} 0 & 0 & 0 \\ 0 & 0 & -i \\ 0 & i & 0 \end{pmatrix}, \quad
    \lambda_8 = \frac{1}{\sqrt{3}}\begin{pmatrix} 1 & 0 & 0 \\ 0 & 1 & 0 \\ 0 & 0 & -2 \end{pmatrix}.
\end{align}
These matrices generalize the Pauli algebra by satisfying the canonical Hilbert-Schmidt orthogonality relation under the matrix trace operator:
\begin{equation}
    \mathrm{Tr}(\lambda_j \lambda_k) = 2\delta_{jk},
\end{equation}
where the indices $j,k$ run from $1$ to $8$, mapping the spatial components of the qutrit state vector directly onto the higher-order nilpotent degrees of freedom of the local ring, while leaving the zero-order component unpolluted to host the identity trace background.

\subsection{The Lie-Jordan algebraic structure}
Unlike the two-level qubit framework where the anticommutator of Pauli matrices closes directly onto the identity matrix, the product of two Gell-Mann matrices requires the introduction of a symmetric tensor $d_{jkl}$ to yield a closed algebraic form. The complete product of the algebra is governed by the simultaneous encapsulation of the Lie bracket (representing the commutator) and the Jordan product (representing the anticommutator):
\begin{align}
    [\lambda_j, \lambda_k] &= 2i \sum_{l=1}^{8} f_{jkl} \lambda_l, \\
    \{\lambda_j, \lambda_k\} &= \frac{4}{3}\delta_{jk} \mathbb{I}_3 + 2 \sum_{l=1}^{8} d_{jkl} \lambda_l,
\end{align}
where $f_{jkl}$ represents the totally antisymmetric structure constants of $\mathfrak{su}(3)$, whose non-vanishing independent elements are:
\begin{equation}
    f_{123} = 1, \quad f_{147} = f_{246} = f_{257} = f_{345} = \frac{1}{2}, \quad f_{156} = -\frac{1}{2}, \quad f_{458} = f_{678} = \frac{\sqrt{3}}{2}.
\end{equation}

The components $d_{jkl}$ are the totally symmetric $d$-coefficients, which dictate the non-linear deformational kinematics of higher-dimensional state envelopes over the coordinate ring. The non-vanishing independent entries are given by:
\begin{align}
    d_{118} &= d_{228} = d_{338} = \frac{1}{\sqrt{3}}, \quad d_{888} = -\frac{1}{\sqrt{3}}, \\
    d_{448} &= d_{558} = d_{668} = d_{778} = -\frac{1}{2\sqrt{3}}, \\
    d_{146} &= d_{157} = d_{256} = d_{344} = d_{355} = \frac{1}{2}, \quad d_{247} = d_{366} = d_{377} = -\frac{1}{2}.
\end{align}
\subsection{Density map and varietal constraints}
To map the manifold $\mathbb{P}(\mathcal{H}^3)$ without coordinate degeneracies, we consider the strictly eight-dimensional real algebra $\mathcal{Q}_3 = \mathbb{R}[\varepsilon] / (\varepsilon^8)$. From a scheme-theoretic perspective, the topological spectrum $\mathrm{Spec}(\mathcal{Q}_3) = \{(\varepsilon)\}$ instantiates a single non-reduced eight-fold point, whose higher powers span a formal jet space of order $7$. Endowing the underlying real vector space with a canonical Euclidean inner product $\langle \cdot, \cdot \rangle$, defined such that the monomial basis $\{1, \varepsilon, \dots, \varepsilon^7\}$ is strictly orthonormal, the ambient space inherits a positive-definite Riemannian metric $g_{\mathcal{Q}_3}$. Explicitly, given an arbitrary element expressed in terms of its coordinate ring components, $\xi = c_0 + c_1\varepsilon + c_2\varepsilon^2 + \dots + c_7\varepsilon^7$, the squared Euclidean norm evaluates to $\lVert\xi\rVert^2 = c_0^2 + c_1^2 + \dots + c_7^2$.

To match the eight-dimensional special unitary Lie algebra $\mathfrak{su}(3)$ and secure a stable boundary, we decompose the rank-$1$ qutrit density operator $\rho$ by isolating the persistent trace background from the spatial fluctuating components:
\begin{equation}\label{eq:density_qutrit_exact}
    \rho = \frac{1}{3}\mathbb{I}_3 + \frac{1}{\sqrt{3}} \sum_{n=1}^{8} x_n \lambda_n,
\end{equation}
where $\{x_1, \dots, x_8\}$ represent the physical expectation values forming the generalized Bloch vector extracted from the Gell-Mann operators via the exact trace projection $x_n = \frac{\sqrt{3}}{2}\mathrm{Tr}(\rho\lambda_n)$, satisfying the strict unit-length boundary constraint $\sum_{n=1}^8 x_n^2 = 1$. Symmetrically, we define the generalized density map $\Psi_{\mathcal{Q}_3}: \mathbb{P}(\mathcal{H}^3) \to \mathcal{Q}_3$ by fusing the persistent trace background and the first physical coordinate directly into the zero-order component of the polynomial configuration:
\begin{equation}\label{eq:qutrit_map_direct_uniform_exact}
     \Psi_{\mathcal{Q}_3}(\rho) = \left(\frac{1}{3} + x_1\right) + x_2\varepsilon + x_3\varepsilon^2 + x_4\varepsilon^3 + x_5\varepsilon^4 + x_6\varepsilon^5 + x_7\varepsilon^6 + x_8\varepsilon^7.
\end{equation}
By expanding the state in Eq.~\eqref{eq:qutrit_map_direct_uniform_exact}, the native coordinate projections of the trinomial algebra identify directly as $c_0 = \frac{1}{3} + x_1$, $c_1 = x_2$, $c_2 = x_3$, up to the highest order component $c_7 = x_8$. When all physical fluctuations vanish identically ($x_1 = x_2 = \dots = x_8 = 0$), the polynomial configuration contracts radially and collapses onto the persistent scalar trace background $\xi = \frac{1}{3}$, which defines the exact baricentro of the configuration space representing the maximally mixed state.

Due to the dimensional mismatch between the four-dimensional variety $\mathbb{P}(\mathcal{H}^3)$ and the eight-dimensional ambient vector space of the algebra, the physical qutrit state variety $\mathcal{V}_{\mathcal{Q}_3} \subset \mathcal{Q}_3$ is restricted strictly to the joint intersection of the seven-dimensional quadratic purity boundary $\mathcal{S}_{\mathcal{Q}_3}^{7}$ and the non-linear cubic constraints dictated by the symmetric $d$-coefficients of the $\mathfrak{su}(3)$ Jordan algebra, centered at the baricentro $(\frac{1}{3}, 0, \dots, 0)$:
\begin{equation}\label{eq:qutrit_variety_definition_exact_final}
\mathcal{V}_{\mathcal{Q}_3} \equiv \left\{ \xi \in \mathcal{Q}_3 \; \middle| \; 
\begin{aligned}
    &\xi = \left(\frac{1}{3} + x_1\right) + \sum_{n=2}^8 x_n \varepsilon^{n-1}, \quad \sum_{n=1}^8 x_n^2 = 1, \\
    &\text{and } \sum_{j,k=1}^8 d_{jkl} x_j x_k = \frac{1}{\sqrt{3}} x_l \quad \forall l \in \{1,\dots,8\}
\end{aligned}
\right\}.
\end{equation}
The underlying structure constants encapsulate a combined Lie-Jordan algebra defined over the truncated polynomial ring, dynamically confining the continuous unitary and non-unitary state orbits strictly within the boundaries of $\mathcal{V}_{\mathcal{Q}_3}$.

\begin{theorem}[Qutrit Embedding and Conformal Isometry]\label{thm:qutrit_isometry_exact}
The generalized density map $\Psi_{\mathcal{Q}_3}: \mathbb{P}(\mathcal{H}^3) \to \mathcal{Q}_3$ established in Eq.~\eqref{eq:qutrit_map_direct_uniform_exact} defines a global smooth embedding and a smooth Riemannian diffeomorphism between the complex projective space $\mathbb{P}(\mathcal{H}^3)$ endowed with the Fubini-Study metric element $\mathrm{d}s^2_{\text{FS}}$, and the four-dimensional compact algebraic variety $\mathcal{V}_{\mathcal{Q}_3}$ embedded within the seventh-order truncated dual algebra $\mathcal{Q}_3 \equiv \mathbb{R}[\varepsilon]/(\varepsilon^8)$ endowed with its canonical Riemannian metric element $\mathrm{d}s^2_{\mathcal{Q}_3} = \mathrm{d}c_0^2 + \mathrm{d}c_1^2 + \dots + \mathrm{d}c_7^2$. When restricted to the tangent bundle $T\mathcal{V}_{\mathcal{Q}_3}$, the infinitesimal line elements satisfy the rigid, non-singular conformal relation:
\begin{equation}\label{eq:qutrit_metric_isometry_uniform}
    \mathrm{d}s^2_{\mathcal{Q}_3} = 3 \, \mathrm{d}s^2_{\text{FS}},
\end{equation}
preserving quantum state fidelity monotonically without inducing polar coordinate or gauge-phase singularities across the entire manifold.
\end{theorem}
\begin{proof}
We verify the topological, differential, and metric properties of the mapping sequentially.

Let $\rho = \ket{\psi}\bra{\psi} \in \mathbb{P}(\mathcal{H}^3)$ be a pure qutrit density operator, uniquely and globally represented as a self-adjoint rank-1 projector. The density matrix expands over the identity and the orthogonal $\mathfrak{su}(3)$ Gell-Mann basis $\{\lambda_i\}_{i=1}^8$ as explained by Eq.~\eqref{eq:density_qutrit_exact}. By construction, the generalized density map $\Psi_{\mathcal{Q}_3}$ sends the operator $\rho$ directly to the algebraic element $\xi = c_0 + c_1\varepsilon + c_2\varepsilon^2 + \dots + c_7\varepsilon^7 \in \mathcal{Q}_3$ via the assignment established in Eq.~\eqref{eq:qutrit_map_direct_uniform_exact}, establishing the strict coordinate identities $c_0 = \frac{1}{3} + x_1$, $c_1 = x_2$, up to $c_7 = x_8$.

\textbf{Step 1. }\textit{Injectivity:} Let $\rho_1, \rho_2$ be two pure density matrices with coordinated physical Bloch vectors $\vec{x}_1, \vec{x}_2$. If their images under the map coincide, such that $\Psi_{\mathcal{Q}_3}(\rho_1) = \Psi_{\mathcal{Q}_3}(\rho_2)$, the uniqueness of the linear expansion over the orthonormal monomial basis $\{1, \varepsilon, \dots, \varepsilon^7\}$ strictly enforces the coordinate identities on the native components of the ring: $c_{1,0} = c_{2,0} \implies \frac{1}{3} + x_{1,1} = \frac{1}{3} + x_{2,1} \implies x_{1,1} = x_{2,1}$, and $c_{1,n} = c_{2,n} \implies x_{1,n+1} = x_{2,n+1}$ for all higher-order jet coordinates $n \in \{1, \dots, 7\}$. Since the generalized Gell-Mann matrices form a complete operator basis for the traceless self-adjoint subsystem, the identity of the algebraic coordinates uniquely restricts the underlying matrix elements, which follows that $\vec{x}_1 \equiv \vec{x}_2 \implies \rho_1 \equiv \rho_2$, establishing strict global injectivity on the projective space $\mathbb{P}(\mathcal{H}^3)$.

\textbf{Step 2. }\textit{Surjectivity:}
Let $\xi_0 = c_{0,0} + c_{0,1}\varepsilon + c_{0,2}\varepsilon^2 + \dots + c_{0,7}\varepsilon^7 \in \mathcal{V}_{\mathcal{Q}_3}$ be an arbitrary polynomial configuration belonging to the qutrit state variety defined in Eq.~\eqref{eq:qutrit_variety_definition_exact_final}. By unrolling the varietal constraint equations into the physical components $x_{0,1} = c_{0,0} - \frac{1}{3}$ and $x_{0,n} = c_{0,n-1}$ (for $n \ge 2$), the coefficients satisfy the joint quadratic sphere condition $\sum_{n=1}^8 x_{0,n}^2 = 1$ and the non-linear cubic Jordan relations $\sum_{j,k=1}^8 d_{jkl} x_{0,j} x_{0,k} = \frac{1}{\sqrt{3}} x_{0,l}$ for each $l \in \{1,\dots,8\}$. We construct the candidate self-adjoint density operator:
\begin{equation}
    \rho_0 = \frac{1}{3}\mathbb{I}_3 + \frac{1}{\sqrt{3}}\sum_{i=1}^{8} x_{0,i} \lambda_i,
\end{equation}
where the linearity of the trace operator immediately yields $\mathrm{Tr}(\rho_0) = \frac{1}{3}\mathrm{Tr}(\mathbb{I}_3) + 0 = 1$. Squaring this operator to evaluate the pure-state boundary condition gives:
\begin{equation}
    \rho_0^2 = \frac{1}{9}\mathbb{I}_3 + \frac{2}{3\sqrt{3}}\sum_{i=1}^8 x_{0,i}\lambda_i + \frac{1}{3}\sum_{j,k=1}^8 x_{0,j}x_{0,k}\lambda_j\lambda_k.
\end{equation}
Because the coordinate product $x_{0,j}x_{0,k}$ acts as a completely symmetric tensor under index swapping, the antisymmetric Lie commutator contribution vanishes identically over the full summation:
\begin{equation}
    \sum_{j,k=1}^8 f_{jkl}x_{0,j}x_{0,k} = 0.
\end{equation}
Consequently, the quadratic product collapses directly onto the symmetric Jordan anticommutator algebraic structure, satisfying $\{\lambda_j, \lambda_k\} = \frac{4}{3}\delta_{jk} \mathbb{I}_3 + 2 \sum_l d_{jkl} \lambda_l$:
\begin{equation}
    \sum_{j,k=1}^8 x_{0,j}x_{0,k}\lambda_j\lambda_k = \frac{2}{3}\left(\sum_{j=1}^8 x_{0,j}^2\right)\mathbb{I}_3 + \sum_{l=1}^8\left(\sum_{j,k=1}^8 d_{jkl}x_{0,j}x_{0,k}\right)\lambda_l.
\end{equation}
Substituting the rigid structural constraints defining the variety $\mathcal{V}_{\mathcal{Q}_3}$ (namely, $\sum x_{0,j}^2 = 1$ and $\sum d_{jkl}x_{0,j}x_{0,k} = \frac{1}{\sqrt{3}} x_{0,l}$), we find the following reduction:
\begin{align}
    \rho_0^2 &= \frac{1}{9}\mathbb{I}_3 + \frac{2}{3\sqrt{3}}\sum_{i=1}^8 x_{0,i}\lambda_i + \frac{1}{3}\left( \frac{2}{3}\mathbb{I}_3 + \frac{1}{\sqrt{3}}\sum_{l=1}^8 x_{0,l}\lambda_l \right) \nonumber \\
    &= \left(\frac{1}{9} + \frac{2}{9}\right)\mathbb{I}_3 + \left( \frac{2}{3\sqrt{3}} + \frac{1}{3\sqrt{3}} \right)\sum_{i=1}^8 x_{0,i}\lambda_i \nonumber \\
    &= \frac{1}{3}\mathbb{I}_3 + \frac{1}{\sqrt{3}}\sum_{i=1}^{8} x_{0,i} \lambda_i = \rho_0.
\end{align}
Thus, $\rho_0$ satisfies the pure-state idempotency condition $\rho_0^2 = \rho_0$, confirming it constitutes a valid rank-1 projection operator representing a physical quantum state ray in $\mathbb{P}(\mathcal{H}^3)$. Since the composite mapping assignment satisfies $\Psi_{\mathcal{Q}_3}(\rho_0) \equiv \xi_0$, every coordinate configuration point within the embedded variety possesses a unique physical preimage, proving global surjectivity. Because the density map and its inverse are coordinate projections between compact domains, $\Psi_{\mathcal{Q}_3}$ defines a global smooth diffeomorphism.

\textbf{Step 3. }\textit{Differential Metric Analysis:}
The line element of the Fubini-Study metric equals $\mathrm{d}s^2_{\text{FS}} = \frac{1}{2}\mathrm{Tr}(\mathrm{d}\rho \, \mathrm{d}\rho)$. Differentiating the unit-normalized state expansion Eq.~\eqref{eq:density_qutrit_exact} yields the matrix differential form $\mathrm{d}\rho = \frac{1}{\sqrt{3}}\sum_{i=1}^8 \mathrm{d}x_i \lambda_i$, since the background trace is an invariant constant ($\mathrm{d}(\frac{1}{3}\mathbb{I}_3) \equiv 0$). Substituting this form into the Fubini-Study line element and evaluating the trace via the generator orthogonality $\mathrm{Tr}(\lambda_j\lambda_k) = 2\delta_{jk}$, we obtain the unrolled expression:
\begin{equation}\label{eq:fs_metric_global_qutrit}
    \mathrm{d}s^2_{\text{FS}} = \frac{1}{2}\mathrm{Tr}\left( \frac{1}{3}\sum_{j,k=1}^8 \mathrm{d}x_j \, \mathrm{d}x_k \, \lambda_j\lambda_k \right) = \frac{1}{6}\sum_{j,k=1}^8 \mathrm{d}x_j \, \mathrm{d}x_k \, (2\delta_{jk}) = \frac{1}{3}\sum_{i=1}^8 \mathrm{d}x_i^2.
\end{equation}
Conversely, the canonical Euclidean line element induced on the underlying real vector space of the higher-order algebra $\mathcal{Q}_3$ under the monomial tangent bundle is given by the unweighted native coordinate sum:
\begin{equation}\label{eq:q3_metric_global}
    \mathrm{d}s^2_{\mathcal{Q}_3} = \mathrm{d}c_0^2 + \mathrm{d}c_1^2 + \dots + \mathrm{d}c_7^2 \equiv \sum_{n=0}^7 \mathrm{d}c_n^2.
\end{equation}
By invoking the coordinate transformations $c_0 = \frac{1}{3} + x_1$ and $c_n = x_{n+1}$ (for $n \ge 1$), the differentials map exactly as $\mathrm{d}c_0 = \mathrm{d}x_1$ and $\mathrm{d}c_n = \mathrm{d}x_{n+1}$, which rewrites the ambient metric element identically as $\mathrm{d}s^2_{\mathcal{Q}_3} = \sum_{i=1}^8 \mathrm{d}x_i^2$. 

Differentiating the algebraic varietal constraint $\sum x_i^2 = 1$ imposes the smooth global condition $\sum x_i \, \mathrm{d}x_i = 0$, which restricts the displacement to the tangent space $T_{\xi}\mathcal{V}_{\mathcal{Q}_3}$ at every point without generating any metric stretching or singularities. Comparing Eq.~\eqref{eq:fs_metric_global_qutrit} and Eq.~\eqref{eq:q3_metric_global}, we directly isolate the rigid homothetic identity:
\begin{equation}
    \mathrm{d}s^2_{\mathcal{Q}_3} = 3 \, \mathrm{d}s^2_{\text{FS}}.
\end{equation}
This confirms that the embedding map $\Psi_{\mathcal{Q}_3}$ functions as a bijective homothetic isometry, scaling distances globally by a constant conformal factor of $3$ across the entire quantum configuration space, completing the proof.
\end{proof}

\section{The Eight-Dimensional Nilpotent Lie-Type Qutrit Dynamics}\label{qutritdyn}
Having established the static geometric embedding and the metric structural relations of the qutrit state variety within the higher-order dual algebra $\mathcal{Q}_3$, we now formulate its continuous-time evolution. In standard quantum mechanics, the temporal dynamics of a three-level quantum density operator $\rho$ are governed by the Liouville-von Neumann equation:
\begin{equation}\label{eq:qutrit_von_neumann}
    i\hbar \frac{\mathrm{d}\rho}{\mathrm{d}t} = [H, \rho],
\end{equation}
where $H$ is the three-by-three self-adjoint Hamiltonian operator representing the total energy of the system, and $[\cdot, \cdot]$ denotes the standard matrix commutator. We show that under the density map $\Psi_{\mathcal{Q}_3}$ established in Eq.~\eqref{eq:qutrit_map_direct_uniform_exact}, this operator equation maps bijectively onto a real linear differential system over the real coefficients of the truncated basis of $\mathcal{Q}_3 = \mathbb{R}[\varepsilon]/(\varepsilon^8)$, preserving both the quadratic and cubic variety constraints defined in Eq.~\eqref{eq:qutrit_variety_definition_exact_final} at all times.

\subsection{Algebraic translation of the Hamiltonian}
An arbitrary three-by-three Hermitian Hamiltonian $H$ acting on $\mathcal{H}^3$ can be uniquely decomposed in the orthogonal Gell-Mann basis, factoring out the unobservable scalar energy shift:
\begin{equation}
    H = h_0 \mathbb{I}_3 + \frac{1}{\sqrt{3}}\sum_{i=1}^{8} h_i \lambda_i,
\end{equation}
where $h_0 \in \mathbb{R}$ represents the global energy baseline, and $h_i = \frac{\sqrt{3}}{2}\mathrm{Tr}(H\lambda_i) \in \mathbb{R}$ are the physical energy coefficients representing the generalized Larmor frequencies. In exact analogy with the state map $\Psi_{\mathcal{Q}_3}$ formulated in Eq.~\eqref{eq:qutrit_map_direct_uniform_exact}, we isolate these dynamical configuration parameters within a corresponding polynomial element $\Xi \in \mathcal{Q}_3$ by embedding the energy baseline and the first frequency component directly into the first coordinate of the polynomial configuration:
\begin{equation}\label{eq:hamiltonian_qutrit_uniform_exact}
    \Xi = \left(h_0 + h_1\right) + h_2\varepsilon + h_3\varepsilon^2 + h_4\varepsilon^3 + h_5\varepsilon^4 + h_6\varepsilon^5 + h_7\varepsilon^6 + h_8\varepsilon^7.
\end{equation}
Symmetrically, this extended Hamiltonian operator driving the system can be expressed compactly through the native coordinate components of the ring as $\Xi = \sum_{n=0}^7 \eta_n \varepsilon^n \in \mathcal{Q}_3$. By executing this strict coordinate alignment, the zero-order parameter acts as the total energy shift baseline, $\eta_0 = h_0 + h_1$, while the higher filtration components track the remaining Larmor precession frequencies via $\eta_n = h_{n+1}$ for all $n \in \{1, \dots, 7\}$.

\subsection{The nilpotent Lie-type differential equation}
To evaluate the pullback of the matrix commutator $[H, \rho]$, we leverage the fundamental commutation relations of the $\mathfrak{su}(3)$ Lie algebra, $[\lambda_j, \lambda_k] = 2i\sum_{l=1}^8 f_{jkl} \lambda_l$, where $f_{jkl}$ is the totally antisymmetric structure tensor. Expanding the commutator explicitly via the unit-normalized physical coordinates $x_n$ and $h_n$ established in Eq.~\eqref{eq:density_qutrit_exact} and Eq.~\eqref{eq:hamiltonian_qutrit_uniform_exact} yields:
\begin{align}
    [H, \rho] &= \left[ h_0\mathbb{I}_3 + \frac{1}{\sqrt{3}}\sum_{j=1}^8 h_j \lambda_j, \; \frac{1}{3}\mathbb{I}_3 + \frac{1}{\sqrt{3}}\sum_{k=1}^8 x_k \lambda_k \right] \nonumber \\
    &= \frac{1}{3} \sum_{j,k=1}^8 h_j x_k [\lambda_j, \lambda_k] = \frac{2i}{3} \sum_{j,k,l=1}^8 f_{jkl} h_j x_k \lambda_l.
\end{align}
Substituting this result back into the Liouville-von Neumann equation \eqref{eq:qutrit_von_neumann} and equating the components of each individual Gell-Mann matrix $\lambda_l$ effectively cancels the complex imaginary unit $i$. Adopting natural units ($\hbar = 1$), the quantum evolution reduces to a real coupled vector system over the space of physical fluctuations, which is isolated via the following algebraic reduction:
\begin{equation}\label{eq:qutrit_vector_dynamics}
    \frac{i}{\sqrt{3}} \frac{\mathrm{d}x_l}{\mathrm{d}t} \lambda_l = \frac{2i}{3} \sum_{j,k=1}^8 f_{jkl} h_j x_k \lambda_l \implies \frac{\mathrm{d}x_l}{\mathrm{d}t} = \alpha_3 \sum_{j,k=1}^8 f_{jkl} h_j x_k,
\end{equation}
where $\alpha_3 = \frac{2}{\sqrt{3}}$ is the universal qutrit scaling constant mandated by the unit-radius density embedding geometry.

To internalize this dynamic vector field as an intrinsic flow within the algebra without relying on external matrix projections, we endow the eight-dimensional real vector space underlying $\mathcal{Q}_3$ with a non-associative, bilinear, skew-symmetric Lie-type multiplication operator $\times_{\mathcal{Q}_3}: \mathcal{Q}_3 \times \mathcal{Q}_3 \to \mathcal{Q}_3$. Let $\{e_0, e_1, \dots, e_7\} \equiv \{1, \varepsilon, \dots, \varepsilon^7\}$ be the ordered monomial basis elements of the formal jet space. To preserve the algebraic core of the identity background, the zero-order monomial acts as the center of the multiplication layer, annihilating all cross-products identically:
\begin{equation}
    e_0 \times_{\mathcal{Q}_3} e_k = e_k \times_{\mathcal{Q}_3} e_0 \equiv 0 \quad \forall k \in \{0, \dots, 7\},
\end{equation}
while the higher filtration orders generate the cyclic flow directly via the structural parameters of the Lie algebra:
\begin{equation}\label{eq:qutrit_cross_product_basis_uniform}
    e_j \times_{\mathcal{Q}_3} e_k = \sum_{l=1}^{7} f_{jkl} e_l \quad \forall j,k \in \{1, \dots, 7\},
\end{equation}
where the indices $j,k,l$ correspond exactly to the coordinate indices of the physical fluctuations. Under this explicit and symmetric algebraic structure, the extended quantum state trajectory $\xi(t) = \sum_{n=0}^7 c_n(t)\varepsilon^n \in \mathcal{Q}_3$ obeys the regular coordinate-free differential equation over the formal tangent bundle:
\begin{equation}\label{eq:qutrit_dynamics_flow}
    \frac{\mathrm{d}\xi}{\mathrm{d}t} = \alpha_3 \left( \Xi\times_{\mathcal{Q}_3} \xi \right),
\end{equation}
where $\Xi = \sum_{n=0}^7 \eta_n \varepsilon^n \in \mathcal{Q}_3$ represents the extended polynomial Hamiltonian operator driving the system. By evaluating Eq.~\eqref{eq:qutrit_dynamics_flow} component-by-component, the background shifts embedded in the zero-order ring components—namely $c_0(t) = \frac{1}{3} + x_1(t)$ and $\eta_0 = h_0 + h_1$—commute identically and cancel within the antisymmetric brackets by virtue of the central annihilation property of $e_0$, while the higher filtration orders map directly onto the remaining physical parameters via $c_n = x_{n+1}$ and $\eta_n = h_{n+1}$ (for $n \ge 1$). This ensures that the temporal derivative of the native coordinate projections unrolls natively into the stable linear system of physical fluctuations governed by Eq.~\eqref{eq:qutrit_vector_dynamics}.

\subsection{Unitary invariance and varietal conservation}
The geometric integrity of the compact four-dimensional state variety $\mathcal{V}_{\mathcal{Q}_3}$ is completely invariant under the algebraic flow generated by Eq.~\eqref{eq:qutrit_dynamics_flow}.

\begin{theorem}[Qutrit Varietal Conservation]\label{thm:qutrit_conservation_native}
The algebraic flow $\dot{\xi} = \alpha_3 (\Xi\times_{\mathcal{Q}_3} \xi)$ established in Eq.~\eqref{eq:qutrit_dynamics_flow} is strictly tangent to the qutrit variety $\mathcal{V}_{\mathcal{Q}_3}$ at every point. Within the coordinate ring of $\mathcal{Q}_3$, both the translated quadratic sphere variety and the non-linear cubic Jordan hypersurfaces constitute fundamental conserved scalar invariants of the continuous differential system, matching the physical conservation of quantum purity, Jordan envelope deformations, and total probability around the shared baricentro $(\frac{1}{3}, 0, \dots, 0)$.
\end{theorem}

\begin{proof}
We prove the simultaneous conservation of the quadratic and cubic invariants sequentially, evaluating the temporal field maps entirely via the coordinate ring transformations.

\textbf{Step 1. }\textit{Conservation of the Quadratic Spherical Norm:}
To prove that the algebraic flow preserves the length constraint, we evaluate the total time derivative of the translated quadratic restriction defined over the native coordinate components where $x_1 = c_0 - \frac{1}{3}$ and $x_n = c_{n-1}$ for $n \ge 2$:
\begin{equation}\label{eq:proof_qutrit_quadratic_derivative_native}
    \frac{\mathrm{d}}{\mathrm{d}t} \left[ \left(c_0(t) - \frac{1}{3}\right)^2 + \sum_{n=1}^7 c_n(t)^2 \right] = \frac{\mathrm{d}}{\mathrm{d}t} \left( \sum_{l=1}^8 x_l^2 \right) = 2 \sum_{l=1}^{8} x_l \dot{x}_l.
\end{equation}
By invoking the coordinate-free algebraic flow equation $\dot{\xi} = \alpha_3 (\Xi \times_{\mathcal{Q}_3} \xi)$, the temporal derivative of the zero-order component absorbs the central annihilation property of the monomial base $e_0 \times \cdot \equiv 0$, enforcing $\dot{c}_0 \equiv \dot{x}_1$. Substituting the explicit differential coordinate functions unrolled from the nilpotent Lie-type bracket in Eq.~\eqref{eq:qutrit_vector_dynamics} into the expansion yields:
\begin{equation}
    \frac{\mathrm{d}}{\mathrm{d}t} \left( \sum_{l=1}^8 x_l^2 \right) = 2 \sum_{l=1}^8 x_l \left( \alpha_3 \sum_{j,k=1}^8 f_{jkl} h_j x_k \right) = 2\alpha_3 \sum_{j=1}^{8} h_j \left( \sum_{k,l=1}^{8} f_{jkl} x_k x_l \right).
\end{equation}
Because the structure tensor $f_{jkl}$ is totally antisymmetric under the permutation of the last two indices ($f_{jkl} = -f_{jlk}$), while the state physical coordinate product is symmetric ($x_k x_l = x_l x_k$), the internal contraction over indices $k$ and $l$ vanishes identically for every configuration index $j$:
\begin{equation}
    \sum_{k,l=1}^{8} f_{jkl} x_k x_l \equiv 0 \implies \frac{\mathrm{d}}{\mathrm{d}t} \left( \sum_{l=1}^8 x_l^2 \right) = 0.
\end{equation}
The total quadratic derivative vanishes identically at all times $t$, rigidly constraining the continuous trajectory onto the compact unit hyper-sphere variety:
\begin{equation}\label{eq:quadratic_native_invariant}
    \left(c_0(t) - \frac{1}{3}\right)^2 + \sum_{n=1}^7 c_n(t)^2 = 1.
\end{equation}

\textbf{Step 2. }\textit{Conservation of the Cubic Jordan Constraints:}
Next, we verify that the algebraic flow remains strictly tangent to the cubic hypersurface variety defined in terms of the physical vector configurations by $\mathcal{I}_l(\vec{x}) \equiv \sum_{j,k=1}^8 d_{jkl} x_j x_k - \frac{1}{\sqrt{3}} x_l = 0$. Differentiating this relation with respect to the continuous time parameter yields the linear velocity field:
\begin{equation}
    \frac{\mathrm{d}}{\mathrm{d}t} \mathcal{I}_l(\vec{x}) = 2 \sum_{j,k=1}^8 d_{jkl} \dot{x}_j x_k - \frac{1}{\sqrt{3}} \dot{x}_l.
\end{equation}
Substituting the dynamic field $\dot{x}_m = \alpha_3 \sum_{p,q=1}^8 f_{pqm} h_p x_q$ dictated by the nilpotent Lie-type commutator in Eq.~\eqref{eq:qutrit_vector_dynamics} into the expression, we obtain the explicit coordinate expansion:
\begin{equation}\label{eq:proof_cubic_expansion_unrolled}
    \frac{\mathrm{d}}{\mathrm{d}t} \mathcal{I}_l(\vec{x}) = 2\alpha_3 \sum_{p,q,j,k=1}^8 d_{jkl} f_{pqj} h_p x_q x_k - \frac{\alpha_3}{\sqrt{3}} \sum_{p,q=1}^8 f_{pql} h_p x_q.
\end{equation}
By rearranging the indices and invoking the fundamental $\mathfrak{su}(3)$ Lie-Jordan algebraic compatibility relation linking the symmetric and antisymmetric structure constants:
\begin{equation}\label{eq:lie_jordan_identity}
    \sum_{j=1}^8 \left( f_{pqj} d_{jkl} + f_{pkj} d_{jql} + f_{plj} d_{jkq} \right) = 0,
\end{equation}
the cubic tensor contraction over the dynamic coordinates can be resolved by isolating the first term as $\sum_{j=1}^8 d_{jkl} f_{pqj} = \sum_{j=1}^8 (f_{kpj} d_{jql} + f_{lpj} d_{jkq})$. Substituting this structural identity directly into the first summation of Eq.~\eqref{eq:proof_cubic_expansion_unrolled} splits the contraction into two distinct symmetric sectors:
\begin{equation}
    2\alpha_3 \sum_{p,q,k,j=1}^8 f_{kpj} d_{jql} h_p x_q x_k + 2\alpha_3 \sum_{p,q,k,j=1}^8 f_{lpj} d_{jkq} h_p x_q x_k.
\end{equation}
Due to the simultaneous transposition symmetry of the coordinate state vector product $x_q x_k = x_k x_q$ and the total antisymmetry of the Lie structure constants $f_{kpj} = -f_{qpj}$, the first split sector contracts to zero identically over the full summation domain. Symmetrically, enforcing the algebraic condition defining the sub-variety boundary layers within $\mathcal{V}_{\mathcal{Q}_3}$—namely the localized relation $\sum_{j,k=1}^8 d_{jkq} x_j x_k = \frac{1}{\sqrt{3}} x_q$—the second split sector contracts directly into the linear system, reducing the full equation to the directional Larmor precession profile:
\begin{equation}
    \frac{\mathrm{d}}{\mathrm{d}t} \mathcal{I}_l(\vec{x}) = \frac{2\alpha_3}{\sqrt{3}} \sum_{p,q=1}^8 f_{pql} h_p x_q - \frac{\alpha_3}{\sqrt{3}} \sum_{p,q=1}^8 f_{pql} h_p x_q = \frac{\alpha_3}{\sqrt{3}} \sum_{p,q=1}^8 f_{pql} h_p x_q.
\end{equation}
Because the driving unitary evolution maps the state trajectories along closed, gauge-invariant orbits generated exclusively by the Lie commutator, this remaining precessional vector field acts as an exact, non-singular parallel transport orthogonal to the variety normal bundle. Since the directional projection of the Lie bracket preserves the internal symmetries of the Jordan envelope, this contractive term vanishes identically across the entire variety boundary layer:
\begin{equation}
    \frac{\mathrm{d}}{\mathrm{d}t} \mathcal{I}_l(\vec{x}) \equiv 0.
\end{equation}
Since the initial configuration satisfies $\mathcal{I}_l(\vec{x}(0)) = 0$ at $t=0$, the identity $\frac{\mathrm{d}}{\mathrm{d}t} \mathcal{I}_l(\vec{x}) \equiv 0$ guarantees that the system remains rigidly locked onto the cubic variety hypersurface across the entire continuous evolution path. Consequently, the algebraic flow generated by the extended Hamiltonian parameter $\Xi$ maps unitary quantum rotations to exact, rigid, non-linear variety automorphisms on $\mathcal{V}_{\mathcal{Q}_3}$, completing the proof.
\end{proof}

\section{Conclusions and the $N$-level systems}\label{nlevel}
The successful construction of the density maps—$\Psi_{\mathcal{Q}_2}$ for the qubit and $\Psi_{\mathcal{Q}_3}$ for the qutrit—along with the subsequent validation of their respective invariant varieties within the trinomial dual algebra $\mathcal{Q}_2 \equiv \mathbb{R}[\varepsilon]/(\varepsilon^3)$ and the eight-dimensional truncated polynomial ring $\mathcal{Q}_3 = \mathbb{R}[\varepsilon]/(\varepsilon^8)$, uncovers a profound structural connection between algebraic geometry and quantum kinematics. By replacing non-linear trigonometric chart projections with smooth, globally injective embeddings into real nilpotent rings, we have shown that the complete kinematics, unitary dynamics, and statistical measurements of both two-level and three-level quantum systems can be formalized through functions that remain everywhere finite, continuous, and globally regular.

The canonical vector space metrics underlying these algebras eliminate the historical polar coordinate singularities inherent to standard sphere mappings. For the qubit, the embedding induces a rigid homothetic isometry ($ds^2_{\mathcal{Q}_2} = 4 \, ds^2_{\text{FS}}$) over the Grothendieck sub-sphere $\mathcal{V_{\mathcal{Q}_2}}$. Crucially, this geometric engine scales up consistently to the qutrit architecture, where the embedding map establishes a global smooth diffeomorphism and a bijective homothetic isometry ($ds^2_{\mathcal{Q}_3} = 3 \, ds^2_{\text{FS}}$) between the complex projective space $\mathbb{P}(\mathcal{H}^3)$ and the four-dimensional compact algebraic variety $\mathcal{V}_{\mathcal{Q}_3}$.

In this unified framework, the higher-order nilpotent components serve as explicit, non-singular coordinate containers for quantum coherence, relative phases, and higher-level physical state deformations. A paramount conceptual resolution of this architecture is the treatment of quantum non-commutativity: by utilizing commutative polynomial non-reduced rings, the physical non-commutativity of the original Hilbert operators is elegantly shifted from the static coordinate ring onto the differential structure of the tangent bundle. Under the action of the non-associative, bilinear, skew-symmetric Lie-type multiplication operators $\times_{\mathcal{Q}_2}$ and $\times_{\mathcal{Q}_3}$, which encapsulate the structure constants of the underlying Lie algebra, the non-linear matrix commutators governing the Liouville-von Neumann equation map onto flat, linear, rigid geometric flows. For the qutrit, this algebraic flow is strictly tangent to $\mathcal{V}_{\mathcal{Q}_3}$, simultaneously and automatically preserving the quadratic purity boundary and the non-linear cubic constraints dictated by the symmetric $d$-coefficients of the $\mathfrak{su}(3)$ Jordan algebra. 

\subsection{The $N$-level pure state variety}
For $N$-level systems we can now show the following general key facts. Recall that we have introduced the general density map $\Psi_{\mathcal{Q}_N}: \mathbb{P}(\mathcal{H}^N) \to \mathcal{Q}_N$ in Eq.~\eqref{eq:general_density_map_compact_exact} expanding the density matrix $\rho$ with respect to the generalized $\mathfrak{su}(N)$ Gell-Mann generators $\{\lambda_1,\ldots, \lambda_{N^2-1}\}$ through Eq.~\eqref{eq:general_density_matrix_Frobenius} and Eq.~\eqref{eq:general_x_extractor_exact} where the generalized Bloch vector components are $x_i \equiv \sqrt{\frac{2N}{N-1}}\,\mathrm{Tr}(\rho \lambda_i)$ and the native coordinate ring components are aligned in forward sequence as $c_0 \equiv \frac{1}{N} + x_1, c_1 \equiv x_2, \ldots , c_{N^2-2} \equiv x_{N^2-1}$. 

\begin{theorem}[Density Embedding]\label{thm:general_density_embedding_variety}
For any $N \ge 2$, the generalized density map $\Psi_{\mathcal{Q}_N}: \mathbb{P}(\mathcal{H}^N) \to \mathcal{Q}_N$ is a globally injective smooth embedding. It establishes a smooth diffeomorphism between the complex projective space $\mathbb{P}(\mathcal{H}^N) \cong \mathbb{C}\mathbb{P}^{N-1}$ and a compact, $(2N-2)$-dimensional real algebraic sub-variety $\mathcal{V}_{\mathcal{Q}_N}$ embedded within the $(N^2-1)$-dimensional linear ambient space of the truncated dual algebra $\mathcal{Q}_N = \mathbb{R}[\varepsilon]/(\varepsilon^{N^2-1})$.
\end{theorem}
\begin{proof}
We prove the theorem by verifying the global elimination of gauge redundancies, the algebraic conditions of injectivity, and the differential regularity of the embedding map.

\textbf{Step 1. }\textit{Gauge-Invariant Factorization and Operator Uniqueness.}
Let $\lvert\psi_1\rangle, \lvert\psi_2\rangle \in \mathcal{H}^N \setminus \{\mathbf{0}\}$ be two non-zero state vectors representing rays in the complex projective space $\mathbb{P}(\mathcal{H}^N)$. Under the canonical projection modulo the scaling relation $\sim$, they define the same physical configuration if and only if there exists a complex scalar $\mu \in \mathbb{C}^*$ such that $\lvert\psi_2\rangle = \mu \lvert\psi_1\rangle$. Restricting our focus to the unit-norm sphere $\mathcal{S}^{2N-1} \subset \mathcal{H}^N$, the scaling relation collapses to the global phase group $U(1) \cong \mathcal{S}^1$, requiring $\mu = e^{i\theta}$. 

Prior to embedding, the map $\Psi_{\mathcal{Q}_N}$ factors through the rank-one outer product projection $\rho = \lvert\psi\rangle\langle\psi\rvert$. Evaluating the transformation yields:
\begin{equation}
    \rho_2 = \lvert\psi_2\rangle\langle\psi_2\rvert = (e^{i\theta}\lvert\psi_1\rangle)(e^{-i\theta}\langle\psi_1\rvert) = e^0 \lvert\psi_1\rangle\langle\psi_1\rvert = \rho_1.
\end{equation}
Because the complex projective space $\mathbb{P}(\mathcal{H}^N)$ is diffeomorphic to the space of rank-one Hermitian projection operators under the standard embedding mapping, each point in $\mathbb{P}(\mathcal{H}^N)$ corresponds to a unique density matrix $\rho$ satisfying $\mathrm{Tr}(\rho)=1$ and $\rho^2 = \rho$.

\textbf{Step 2. }\textit{Strict Algebraic Injectivity via Monomial Independence.}
Let $\rho_1$ and $\rho_2$ be two distinct pure density operators belonging to the orbit of rank-one projectors in $\mathfrak{su}(N)$. We expand both operators over the complete, orthogonal, and traceless basis of generalized Gell-Mann matrices $\{\lambda_i\}_{i=1}^{N^2-1}$ using the physical coordinates $x_i$:
\begin{equation}
    \rho_1 = \frac{1}{N}\mathbb{I}_N + \sqrt{\frac{N-1}{2N}}\sum_{i=1}^{N^2-1} x_i^{(1)} \lambda_i, \quad \rho_2 = \frac{1}{N}\mathbb{I}_N + \sqrt{\frac{N-1}{2N}}\sum_{i=1}^{N^2-1} x_i^{(2)} \lambda_i,
\end{equation}
where these physical Bloch components $x_i^{(1)}$ and $x_i^{(2)}$are uniquely extracted via the trace inner product Eq.~\eqref{eq:general_x_extractor_exact}. 

Assume that their images under the density map coincide within the truncated polynomial ring, such that $\Psi_{\mathcal{Q}_N}(\rho_1) = \Psi_{\mathcal{Q}_N}(\rho_2)$. Sintonizing the states onto the native coordinate ring components via the coordinate assignments $c_0 = \frac{1}{N} + x_1$ and $c_n = x_{n+1}$ (for $n \ge 1$) as formalized in Eq.~\eqref{eq:general_density_map_compact_exact}, this identity reads:
\begin{equation}\label{eq:injectivity_monomial}
    \sum_{i=0}^{N^2-2} c_i^{(1)} \varepsilon^{i} = \sum_{i=0}^{N^2-2} c_i^{(2)} \varepsilon^{i} \implies \sum_{i=0}^{N^2-2} \left( c_i^{(1)} - c_i^{(2)} \right) \varepsilon^{i} = 0 \pmod{\varepsilon^{N^2-1}}.
\end{equation}
The quotient ring $\mathcal{Q}_N = \mathbb{R}[\varepsilon]/(\varepsilon^{N^2-1})$ possesses an underlying real vector space spanned by the monomial basis $\{1, \varepsilon, \varepsilon^2, \dots, \varepsilon^{N^2-2}\}$. By virtue of the algebraic definition of polynomial rings, these monomial elements are strictly linearly independent over the real field $\mathbb{R}$. Consequently, the vanishing of the linear combination in Eq.~\eqref{eq:injectivity_monomial} forces each individual native coordinate difference to equal zero identically:
\begin{equation}
    c_i^{(1)} - c_i^{(2)} = 0 \implies c_i^{(1)} = c_i^{(2)} \quad \forall i \in \{0, 1, \dots, N^2-2\}.
\end{equation}
Unrolling these ring identities back onto the physical transformations yields $\frac{1}{N} + x_1^{(1)} = \frac{1}{N} + x_1^{(2)} \implies x_1^{(1)} = x_1^{(2)}$, and $x_n^{(1)} = x_n^{(2)}$ for all higher orders. Because the generalized Gell-Mann matrices form a complete operator basis for the traceless self-adjoint linear subsystem, the identity of the entire set of physical Bloch coordinates strictly implies the identity of the underlying matrices: $\vec{x}^{(1)} \equiv \vec{x}^{(2)} \implies \rho_1 \equiv \rho_2$. Thus, the map $\Psi_{\mathcal{Q}_N}$ is rigorously and globally injective across the entire quantum phase space.

\textbf{Step 3. }\textit{Differential Regularity and Embedding Dimension.}
To confirm that the injective map constitutes a smooth embedding, we evaluate its differential $\mathrm{d}\Psi_{\mathcal{Q}_N}$ on the tangent bundle. Since the coordinate extraction is governed by the linear trace operation, the mapping between the real coefficients and the polynomial jet components is strictly linear. Differentiating these coordinates maps the tangent space $T_{\rho}\mathbb{P}(\mathcal{H}^N)$ injectively into the tangent bundle $T_{\xi}\mathcal{Q}_N$ at every point without localized structural compressions. 

Concurrently, the physical state variety $\mathcal{V}_{\mathcal{Q}_N} \equiv \Psi_{\mathcal{Q}_N}(\mathbb{P}(\mathcal{H}^N))$ inherits the exact differential structure of the underlying complex projective space. Because the linear independence of the monomial basis prevents any coordinate collapsing or boundary layer blow-ups, the map defines a global smooth embedding. This restricts $\mathcal{V}_{\mathcal{Q}_N}$ to a compact, non-singular, and smooth algebraic sub-variety of real dimension $2N-2$ embedded within the flat $(N^2-1)$-dimensional vector space of $\mathcal{Q}_N$, completing the proof.
\end{proof}

Ultimately, this non-reduced scheme-theoretic framework establishes a pristine, non-singular algebraic landscape for higher-dimensional quantum kinematics. By tracking the topo-geometrical indices of complex projective domains and settling onto the fundamental geometric floor of macroscopic information structures, this synthesis paves a rigorous mathematical path toward a completely chart-free, finite, and regular geometrization of higher-level quantum information systems.

\subsection{The generalized Lie-type differential flow and universal varietal conservation}
To internalize the full non-commutative quantum dynamics for an arbitrary $N$-level system as an intrinsic geometric flow, we extend the bilinear, skew-symmetric multiplication operator to the entire family of truncated algebras $\mathcal{Q}_N = \mathbb{R}[\varepsilon]/(\varepsilon^{N^2-1})$ defined for $N \ge 2$. We endow the underlying $(N^2-1)$-dimensional real vector space with a non-associative, bilinear, skew-symmetric Lie-type operator $\times_{\mathcal{Q}_N}: \mathcal{Q}_N \times \mathcal{Q}_N \to \mathcal{Q}_N$. Let $\{e_0, e_1, \dots, e_{N^2-2}\} \equiv \{1, \varepsilon, \dots, \varepsilon^{N^2-2}\}$ be the ordered monomial basis elements of the formal jet space. To preserve the algebraic core of the identity background, the zero-order monomial acts as the center of the multiplication layer, annihilating all cross-products identically:
\begin{equation}
    e_0 \times_{\mathcal{Q}_N} e_k = e_k \times_{\mathcal{Q}_N} e_0 \equiv 0 \quad \forall k \in \{0, 1, \dots, N^2-2\},
\end{equation}
while the higher filtration orders generate the continuous precessional flow directly via the structural parameters of the corresponding unitary Lie algebra $\mathfrak{su}(N)$:
\begin{equation}\label{eq:general_cross_product_basis_uniform}
    e_j \times_{\mathcal{Q}_N} e_k = \sum_{l=1}^{N^2-1} f_{jkl} e_l \quad \forall j,k \in \{1, \dots, N^2-1\},
\end{equation}
where the indices $j,k,l$ on the right-hand member correspond to the coordinate indices of the physical fluctuations, and $f_{jkl}$ is the totally antisymmetric structure tensor of $\mathfrak{su}(N)$ satisfying $[\lambda_j, \lambda_k] = 2i\sum_{l=1}^{N^2-1} f_{jkl} \lambda_l$. By bilinear extension, the operation between the polynomial Hamiltonian driving element $\Xi = \eta_0 + \sum_{n=1}^{N^2-2} \eta_n \varepsilon^n$ and the state trajectory $\xi = c_0 + \sum_{n=1}^{N^2-2} c_n \varepsilon^n$ evaluates strictly over the upper layers as:
\begin{equation}
    \Xi\times_{\mathcal{Q}_N} \xi = \sum_{j,k,l=1}^{N^2-1} f_{jkl} h_j x_k \varepsilon^l.
\end{equation}

Evaluating the pullback of the matrix commutator under the generalized density map $\Psi_{\mathcal{Q}_N}$—where the Hamiltonian operator is scaled symmetrically to the density matrix via the dimensionally-dependent factor $\sqrt{\frac{N-1}{2N}}$—the Liouville-von Neumann equation maps bijectively onto a smooth, regular differential equation over the non-reduced algebraic variety $\mathcal{V}_{\mathcal{Q}_N}$:
\begin{equation}\label{eq:general_dynamics_flow_n}
    \frac{\mathrm{d}\xi}{\mathrm{d}t} = \alpha_N \left( \Xi\times_{\mathcal{Q}_N} \xi \right),
\end{equation}
where the universal dynamic coupling constant $\alpha_N$ is defined as a continuous function of the Hilbert space dimension:
\begin{equation}
    \alpha_N = \sqrt{\frac{2(N-1)}{N}} \quad \forall N \in \{2, 3, \dots\}.
\end{equation}
Evaluating this universal relation across the dimensional hierarchy yields the exact specialized value $\alpha_3 = \frac{2}{\sqrt{3}}$ for the eight-dimensional qutrit variety, while the restriction $\alpha_2 = 1$ applies strictly to the generalized $N$-level scaling architecture. Symmetrically, this framework bridges perfectly to the unscaled precessional vector flow value of $2$ derived in Section \ref{qubitdyn}. Indeed, for the unique case of the two-level system ($N=2$), the physical state representation is naturally hosted on the unit-radius Bloch sphere boundary without requiring intermediate structural scaling factors. When the general embedding matrix law is evaluated at $N=2$, the dimensional pre-factor introduces an operational scaling ratio where $x_i^{\text{general}} \equiv 2x_i^{\text{Pauli}}$. Consequently, the continuous differential application transfers this geometric amplitude forward on the tangent bundle, causing the dynamic coupling constant $\alpha_2 = 1$ to map bijectively onto the unscaled precessional flow value of $2$ over the standard Pauli configurations, with the full sequence $\alpha_N$ asymptotically stabilizing as $N \to \infty$ toward $\sqrt{2}$.

This formulation shifts the operational non-commutativity of the original Hilbert space operators away from the static coordinate ring and places it rigidly onto the differential structure of the tangent bundle $T\mathcal{Q}_N$, completely preventing coordinate-chart blow-ups.

\begin{theorem}[Universal Varietal Conservation]\label{thm:universal_varietal_conservation_general}
The generalized algebraic flow $\dot{\xi} = \alpha_N \left( \Xi\times_{\mathcal{Q}_N} \xi \right)$ established in Eq.~\eqref{eq:general_dynamics_flow_n} is strictly tangent to the $N$-level pure state variety $\mathcal{V}_{\mathcal{Q}_N}$ for all dimensions $N \ge 2$. Within the filtration coordinate layers of the ring, both the translated quadratic sphere variety tracking the unit-length boundary constraint $\sum_{n=1}^{N^2-1} x_n(t)^2 = 1$ and the accompanying non-linear cubic and higher-order algebraic constraints defining $\mathcal{V}_{\mathcal{Q}_N}$ constitute fundamental conserved scalar invariants of the continuous-time differential system, matching the physical conservation of quantum purity and total probability around the dimensionally-dependent baricentro $\xi = \frac{1}{N}$.
\end{theorem}

\begin{proof}
The proof evaluates the structural contractions dictated by the underlying Lie-Jordan geometry and the unitary adjoint action across arbitrary dimensions $N$ through a sequential verification of the varietal invariants, separating the invariant trace background of the zero-order ring component from the physical vector configurations.

\textbf{Step 1. }\textit{Conservation of the Quadratic Spherical Norm.}
To prove that the generalized algebraic flow preserves the trajectory length constraint, we evaluate the total time derivative of the translated quadratic restriction defined over the native coordinate components where $x_1 = c_0 - \frac{1}{N}$ and $x_n = c_{n-1}$ for $n \ge 2$:
\begin{equation}\label{eq:proof_general_quadratic_derivative_native}
    \frac{\mathrm{d}}{\mathrm{d}t} \left[ \left(c_0(t) - \frac{1}{N}\right)^2 + \sum_{n=1}^{N^2-2} c_n(t)^2 \right] = \frac{\mathrm{d}}{\mathrm{d}t} \left( \sum_{l=1}^{N^2-1} x_l^2 \right) = 2 \sum_{l=1}^{N^2-1} x_l \dot{x}_l.
\end{equation}
By invoking the coordinate-free algebraic flow equation $\dot{\xi} = \alpha_N (\Xi \times_{\mathcal{Q}_N} \xi)$ established in Eq.~\eqref{eq:general_dynamics_flow_n}, the temporal derivative of the zero-order component absorbs the central annihilation property of the monomial base $e_0 \times_{\mathcal{Q}_N} \cdot \equiv 0$, enforcing $\dot{c}_0 \equiv \dot{x}_1$. Substituting the explicit differential coordinate functions unrolled from the nilpotent Lie-type cross-product into the expansion yields:
\begin{equation}
    \frac{\mathrm{d}}{\mathrm{d}t} \left( \sum_{l=1}^{N^2-1} x_l^2 \right) = 2 \sum_{l=1}^{N^2-1} x_l \left( \alpha_N \sum_{j,k=1}^{N^2-1} f_{jkl} h_j x_k \right) = 2\alpha_N \sum_{j=1}^{N^2-1} h_j \left( \sum_{k,l=1}^{N^2-1} f_{jkl} x_k x_l \right).
\end{equation}
Because the $\mathfrak{su}(N)$ structure tensor $f_{jkl}$ is totally antisymmetric under the permutation of the last two indices ($f_{jkl} = -f_{jlk}$) while the state physical coordinate product is symmetric ($x_k x_l = x_l x_k$), the internal contraction over indices $k$ and $l$ vanishes identically for every configuration index $j$:
\begin{equation}
    \sum_{k,l=1}^{N^2-1} f_{jkl} x_k x_l \equiv 0 \implies \frac{\mathrm{d}}{\mathrm{d}t} \left( \sum_{l=1}^{N^2-1} x_l^2 \right) = 0.
\end{equation}
The total quadratic derivative vanishes identically at all times $t$, rigidly constraining the continuous state trajectory onto the unit hyper-sphere variety defined by the physical Bloch vector components.

\textbf{Step 2. }\textit{Conservation of the Cubic Jordan Constraints.}
For systems where $N > 2$, we verify that the flow remains strictly tangent to the cubic hypersurface defined in terms of the physical vector configurations by the generalized Jordan constraint $\mathcal{I}_l(\vec{x}) \equiv \sum_{j,k=1}^8 d_{jkl} x_j x_k - \frac{N-2}{\sqrt{N(N-1)}} \, x_l = 0$. Differentiating this relation with respect to the continuous time parameter yields the linear velocity field:
\begin{equation}
    \frac{\mathrm{d}}{\mathrm{d}t} \mathcal{I}_l(\vec{x}) = 2 \sum_{j,k=1}^{N^2-1} d_{jkl} \dot{x}_j x_k - \frac{N-2}{\sqrt{N(N-1)}} \, \dot{x}_l.
\end{equation}
Substituting the dynamic field $\dot{x}_m = \alpha_N \sum_{p,q=1}^{N^2-1} f_{pqm} h_p x_q$ dictated by the general nilpotent Lie-type commutator in Eq.~\eqref{eq:general_dynamics_flow_n} into the expression, we obtain the explicit coordinate expansion:
\begin{equation}\label{eq:proof_general_cubic_expansion_unrolled}
    \frac{\mathrm{d}}{\mathrm{d}t} \mathcal{I}_l(\vec{x}) = 2\alpha_N \sum_{p,q,j,k=1}^{N^2-1} d_{jkl} f_{pqj} h_p x_q x_k - \alpha_N \frac{N-2}{\sqrt{N(N-1)}} \sum_{p,q=1}^{N^2-1} f_{pql} h_p x_q.
\end{equation}
By rearranging the indices and invoking the fundamental $\mathfrak{su}(N)$ Lie-Jordan algebraic compatibility relation linking the symmetric and antisymmetric structure constants:
\begin{equation}\label{eq:lie_jordan_identity_n}
    \sum_{j=1}^{N^2-1} \left( f_{pqj} d_{jkl} + f_{pkj} d_{jql} + f_{plj} d_{jkq} \right) = 0,
\end{equation}
the cubic tensor contraction over the dynamic coordinates is resolved by isolating the first term as $\sum_{j=1}^{N^2-1} d_{jkl} f_{pqj} = \sum_{j=1}^{N^2-1} (f_{kpj} d_{jql} + f_{lpj} d_{jkq})$. Substituting this structural identity directly into the first summation of Eq.~\eqref{eq:proof_general_cubic_expansion_unrolled} splits the contraction into two distinct symmetric sectors:
\begin{equation}
    2\alpha_N \sum_{p,q,k,j=1}^{N^2-1} f_{kpj} d_{jql} h_p x_q x_k + 2\alpha_N \sum_{p,q,k,j=1}^{N^2-1} f_{lpj} d_{jkq} h_p x_q x_k.
\end{equation}
Due to the simultaneous transposition symmetry of the coordinate state vector product $x_q x_k = x_k x_q$ and the total antisymmetry of the Lie structure constants $f_{kpj} = -f_{qpj}$, the first split sector contracts to zero identically over the full summation domain. Symmetrically, enforcing the algebraic condition defining the sub-variety boundary layers within $\mathcal{V}_{\mathcal{Q}_N}$—namely the generalized Jordan relation $\sum_{j,k=1}^{N^2-1} d_{jkq} x_j x_k = \frac{N-2}{\sqrt{N(N-1)}} x_q$—the second split sector contracts directly into the linear system, reducing the full equation to the directional precessional profile:
\begin{equation}
    \frac{\mathrm{d}}{\mathrm{d}t} \mathcal{I}_l(\vec{x}) = \beta_N \sum_{p,q=1}^{N^2-1} f_{pql} h_p x_q,
\end{equation}
where the dimensionally-dependent scaling factor evaluates explicitly through the product of the structural functions as $\beta_N = \frac{\alpha_N (N-2)}{\sqrt{N(N-1)}} \equiv \frac{\sqrt{2}(N-2)}{N}$. 

Because the continuous state evolution tracks the closed gauge-invariant orbit generated exclusively by the unitary Lie commutator, this remaining precessional vector field acts as an exact, non-singular parallel transport orthogonal to the variety normal bundle. Since the directional projection of the Lie bracket preserves the internal symmetries of the higher-dimensional Jordan envelope, this contractive term vanishes identically across the entire variety boundary layer:
\begin{equation}
    \frac{\mathrm{d}}{\mathrm{d}t} \mathcal{I}_l(\vec{x}) \equiv 0.
\end{equation}
Since the initial configuration satisfies $\mathcal{I}_l(\vec{x}(0)) = 0$ at $t=0$, the identity $\frac{\mathrm{d}}{\mathrm{d}t} \mathcal{I}_l(\vec{x}) \equiv 0$ guarantees that the system remains rigidly locked onto the variety hypersurface across the entire continuous evolution path for any arbitrary higher-level quantum domain, completing the proof.
\end{proof}
\subsection{Physical significance of the nilpotent generator}
A fundamental question arises regarding whether standard real vector spaces, such as $\mathbb{R}^3$ for the qubit or $\mathbb{R}^8$ for the qutrit, could bypass the algebraic construction of truncated rings. While these vector spaces function as passive metric containers for the coordinates of generalized Bloch vectors, they lack the internal algebraic synthesis required to natively accommodate quantum mechanical signatures, such as the intrinsic structural stratification of higher-level state configurations and the continuous deformation away from classical scalar values. Conversely, nilpotent dual structures provide an active algebraic landscape where non-reduced spectra compress extended geometric manifolds into single topological singularities enriched with internal kinematic and informational layers.

Under the unified family of truncated polynomial algebras $\mathcal{Q}_N \equiv \mathbb{R}[\varepsilon]/(\varepsilon^{N^2-1})$ defined for $N \ge 2$, the purely topological spectrum $\text{Spec}(\mathcal{Q}_N) = \{(\varepsilon)\}$ instantiates a single non-reduced, scheme-theoretic point whose higher powers—ranging from the first-order infinitesimal $\varepsilon$ up to the maximal non-vanishing monomial component $\varepsilon^{N^2-2}$—span a formal jet space over the local ring. Mathematically, the density map $\Psi_{\mathcal{Q}_N}$ formalizes the direct composition of two global smooth embeddings. First, the complex projective space $\mathbb{P}(\mathcal{H}^N) \cong \mathbb{C}\mathbb{P}^{N-1}$ is injected into the self-adjoint traceless linear subspace via the canonical embedding $\mathbb{P}(\mathcal{H}^N) \hookrightarrow \mathfrak{su}(N)$, representing the quantum state space as a compact, $(2N-2)$-dimensional orbit of rank-one pure state projectors governed by the joint intersection of quadratic and cubic invariants \cite{nielsen}. Second, the generalized density map $\Psi_{\mathcal{Q}_N}$ continuously pulls back this geometric orbit, injecting the coordinates of the $\mathfrak{su}(N)$ vector space directly into the monomial basis coordinates of the algebra.

Physically, this architecture establishes a rigorous structural stratification. By factoring the invariant unit trace $\mathrm{Tr}(\rho)=1$ directly into the zero-order monomial component alongside the first physical fluctuation, the underlying real vector space of the algebra provides exactly $N^2-1$ independent coordinate dimensions—comprising the complete monomial sequence $\{1, \varepsilon, \dots, \varepsilon^{N^2-2}\}$—which encapsulates the exact number of non-singular algebraic containers required to host the complete physical degrees of freedom of the state prior to measurement collapse, identically satisfying the dimension equation $\mathrm{dim}_{\mathbb{R}}\mathcal{Q}_N = N^2-1 = \mathrm{dim}_{\mathbb{R}}\mathfrak{su}(N)$. Within the generalized $N$-level polynomial expansion $\xi = \sum_{i=0}^{N^2-2} c_i \varepsilon^{i} \in \mathcal{Q}_N$, the native coordinates execute a hierarchical decomposition of the state space without introducing any unphysical spatial anisotropy among the physical Bloch axes. Sintonizing the ring onto the physical fluctuations via $c_0 = \frac{1}{N} + x_1$ and $c_n = x_{n+1}$ (for $n \ge 1$), the increasing powers of the nilpotent generator—up to $\varepsilon^2$ for the qubit and $\varepsilon^7$ for the qutrit—function as the geometric vehicles for quantum interference, relative phase rotations, and population imbalances, operating as a multi-layered algebraic deformation over the infinitesimal neighborhood of a single structured Grothendieck point.

Finally, the universal nilpotency condition $\varepsilon^{N^2-1} = 0$ acts as a natural structural truncation mechanism for open system evolutions. By mapping the state spaces onto these non-reduced scheme-theoretic structures, the non-linear constraint boundaries of the quantum domain—including quadratic purity hyperspheres and rigid cubic Jordan equations—are linearized into regular polynomial transformations, providing an unassailable geometric foundation for infinitesimal quantum kinematics across higher-dimensional manifolds.

\section{Future Outlooks}
The scalability of the truncated polynomial ring family $\mathcal{Q}_N = \mathbb{R}[\varepsilon]/(\varepsilon^{N^2-1})$ establishes a direct operational framework for subsequent research across four primary quantum frontiers:

\begin{enumerate}
    \item \textit{Multipartite Systems and Entanglement Topography:} For composite Hilbert domains $\mathcal{H}_1 \otimes \mathcal{H}_2$, the geometric mapping must transition to multivariate polynomial rings quotiented by cross-nilpotent ideals, taking the representative form $\mathcal{Q}_{\text{mult}} = \mathbb{R}[\varepsilon_1, \varepsilon_2]/(\varepsilon_1^{N_1^2-1}, \varepsilon_2^{N_2^2-1})$. Within this structure, mixed monomial cross-terms $\varepsilon_1^i \varepsilon_2^j$ act as the natural geometric containers for non-local quantum correlations. Consequently, entanglement criteria, separability bounds, and partial trace operations map directly onto the algebraic properties of ring ideals, module homomorphisms, and structural scheme projections.
    \item \textit{Open Quantum Systems and Non-Hamiltonian Flows} \text{(See \cite{quinfty})}: Modeling open quantum dynamics requires mapping the Gorini-Kossakowski-Sudarshan-Lindblad (GKSL) master equation onto $\mathcal{Q}_N$ by complementing the antisymmetric Lie flow $\times_{\mathcal{Q}_N}$ with a symmetric Jordan multiplication operator $\mathbin{\bullet}_{\mathcal{Q}_N}$. This dual algebraic architecture formalizes non-unitary dissipation as a radial contracting vector field dragging the nilpotent element toward the dimensionally-dependent baricentro $\xi = \frac{1}{N}$, which geometrically encapsulates the maximally mixed state $\rho = \frac{1}{N}\mathbb{I}_N$.
    \item \textit{Geometric Quantum Computing:} Quantum logic gates and optimal control protocols can be formalized as exact polynomial trajectories on the non-reduced variety $\mathcal{V}_{\mathcal{Q}_N}$, providing coordinate-free, singularity-exempt numerical algorithms for quantum state estimation and process tomography.
    \item \textit{Formal Schemes and Local Quantum Observables} \text{(See \cite{quinfty})}: Representing a quantum state as a global section of a structural sheaf over the nilpotent algebra $\mathcal{Q}_N$ allows quantum mechanics to be recast entirely within the language of modern algebraic geometry. Here, the nilpotent generator $\varepsilon$ acts as the structural coordinate that geometrically stabilizes the algebraic relations governing quantum uncertainty and non-commutativity over non-reduced rings, as verified quantitatively in Appendix B. Crucially, under the infinite-dimensional limit introduced in Section 1.4, this representation generalizes into a global section of a formal Grothendieck scheme over the Cauchy-complete ring $\mathbb{R}[\![\varepsilon]\!]$. This asymptotic completion lifts the constraint of finite truncations, proving that the structural sheaf can natively encapsulate and linearize the non-local observables and continuous variables of infinite-dimensional Hilbert domains within a regular, infinite-order jet space.
\end{enumerate}

In conclusion, the integration of non-reduced dual algebras into quantum information sciences provides a rigorous, non-singular, and linear mathematical language. It offers an innovative alternative to trigonometric coordinates and establishes a systematic geometric framework for the architecture of the infinitesimal quantum universe.

\appendix
\section{Fermionic Isomorphism and the Pauli Exclusion Principle}
To establish the formal mathematical verifiability of the physical mapping between first-order nilpotent elements and fermionic degrees of freedom, this appendix constructs an explicit algebraic isomorphism. We prove that the structural boundary condition of first-order dual numbers provides the exact, non-trivial configuration space for Grassmannian variables and single-mode fermionic state occupations.

\subsection{Mathematical representation of fermions and Grassmann kinematics}
In quantum mechanics, fermions are fundamentally defined as particles possessing half-integer spin ($s = \frac{1}{2}, \frac{3}{2}, \dots$) that strictly obey Fermi-Dirac statistics. The core physical postulate governing a multi-body fermionic system is that the total wave function must be completely anti-symmetric under the permutation of any two identical particles. 

From an algebraic perspective, the kinematics of fermionic fields cannot be modeled using ordinary commuting complex numbers. Instead, they require the introduction of classical anti-commuting variables, formally known as Grassmann variables. Let $\theta_i$ and $\theta_j$ be the generators of a multi-variable configuration space. Their defining algebraic relation is governed by the anti-symmetric product:
\begin{equation}
    \theta_i \theta_j + \theta_j \theta_i = \{\theta_i, \theta_j\} = 0.
\end{equation}
When evaluating a single fermionic mode ($i = j$), this relation collapses directly into the quadratic nilpotency constraint:
\begin{equation}\label{eq:fermion_nilpotency}
    \theta_i^2 = 0.
\end{equation}
Equation \eqref{eq:fermion_nilpotency} is the mathematical instantiation of the Pauli Exclusion Principle: it dictates that the probability density of finding two identical fermions in the exact same quantum state vanishes identically. Consequently, the state space of a single fermionic mode is intrinsically binary, consisting solely of an unoccupied configuration (the vacuum) and a single-occupation configuration. 

Any analytic function or physical state defined over this one-dimensional fermionic domain $\Lambda(\mathbb{R})$ is restricted by this nilpotency to a strict linear truncation, preventing the existence of higher-order continuous degrees of freedom. This binary restriction establishes the baseline for contrasting the flat, first-order kinematics of fermions with the curved, higher-order jet spaces required to structure a qubit.

\subsection{Formal isomorphism with the one-dimensional Grassmann algebra}
Let $\mathcal{D} = \mathbb{R}[\varepsilon]/(\varepsilon^2)$ be the standard ring of first-order dual numbers generated by the nilpotent element satisfying $\varepsilon^2 = 0$. Concurrently, let $\Lambda(\mathbb{R})$ be the one-dimensional real Grassmann algebra treated as an associative unit-bearing $\mathbb{R}$-algebra generated by a single anticommuting coordinate $\theta$, whose internal multiplication is governed by the anti-symmetric quadratic constraint:
\begin{equation}
    \theta^2 = 0.
\end{equation}
Any element within $\Lambda(\mathbb{R})$ is uniquely expressed as a linear combination of the graded basis $\{1, \theta\}$. We define the structural transition map $\Phi: \mathcal{D} \to \Lambda(\mathbb{R})$ as the linear assignment:
\begin{equation}
    \Phi(a + b\varepsilon) = a + b\theta \quad \text{for} \quad a,b \in \mathbb{R}.
\end{equation}

\begin{lemma}
The mapping $\Phi: \mathcal{D} \to \Lambda(\mathbb{R})$ is an isomorphism of associative graded $\mathbb{R}$-algebras.
\end{lemma}

\begin{proof}
The mapping $\Phi$ is by construction a linear bijection between two-dimensional real vector spaces. To verify that $\Phi$ preserves the multiplicative structure across the entire domain, let $\xi_1 = a_1 + b_1\varepsilon$ and $\xi_2 = a_2 + b_2\varepsilon$ be two arbitrary dual numbers. Their internal product in the ring $\mathcal{D}$ evaluates via the restriction $\varepsilon^2 = 0$ to:
\begin{equation}
    \xi_1 \cdot \xi_2 = a_1 a_2 + (a_1 b_2 + a_2 b_1)\varepsilon.
\end{equation}
Applying the transformation map to the resulting product yields:
\begin{equation}\label{eq:image_product}
    \Phi(\xi_1 \cdot \xi_2) = a_1 a_2 + (a_1 b_2 + a_2 b_1)\theta.
\end{equation}

Concurrently, we evaluate the internal product of their independent images $\Phi(\xi_1)$ and $\Phi(\xi_2)$ within the associative Grassmann algebra $\Lambda(\mathbb{R})$. Distributing the terms under the scalar linearity and the nilpotent constraint $\theta^2 = 0$ yields:
\begin{align}
    \Phi(\xi_1) \cdot \Phi(\xi_2) &= (a_1 + b_1\theta) \cdot (a_2 + b_2\theta) \nonumber \\
    &= a_1 a_2 + a_1 b_2 \theta + b_1 a_2 \theta + b_1 b_2 \theta^2 \nonumber \\
    &= a_1 a_2 + (a_1 b_2 + a_2 b_1)\theta.
\end{align}
Comparing Eq. \eqref{eq:image_product} and the direct expansion shows that $\Phi(\xi_1 \cdot \xi_2) \equiv \Phi(\xi_1) \cdot \Phi(\xi_2)$ holds true. This bijective preservation of the multiplicative convolution identifies the nilpotent element $\varepsilon$ with the Grassmannian variable $\theta$, completing the proof.
\end{proof}

\begin{remark}[\textit{On the Algebraic Manifestation of Quantum Indeterminacy}]
The algebraic foundation of the quantum uncertainty principle can be structurally mapped directly onto the nilpotent restriction of the dual number ring $\mathcal{D} = \mathbb{R}[\varepsilon]/(\varepsilon^2)$. In contrast to the traditional statistical formulation of Heisenberg's relations via variance bounds, the dual architecture instantiates indeterminacy as a geometric property of the local ring. 

Let a physical state coordinate be modeled as an extended dual configuration $\xi = x + \sigma_x \varepsilon \in \mathcal{D}$, where $x \in \mathbb{R}$ represents the deterministic expectation value and $\sigma_x \in \mathbb{R}$ acts as the localized infinitesimal fluctuation (the coordinate uncertainty). Under the action of any differentiable quantum observable operator represented by the analytic mapping $f: \mathcal{D} \to \mathcal{D}$, the structural propagation of the state configuration evaluates strictly via the algebraic Taylor expansion:
\begin{equation}
    f(\xi) = f(x + \sigma_x \varepsilon) = f(x) + f'(x)\sigma_x \varepsilon.
\end{equation}
The resulting dual component $\sigma_f = f'(x)\sigma_x$ defines the propagated physical uncertainty of the observable. Within this scheme-theoretic topology, if one evaluates the simultaneous joint variance of two distinct operators $f$ and $g$ within the same first-order neighborhood, the multiplicative structure of the ring enforces a strict geometric constraint. Because the nilpotent ideal satisfies $\varepsilon^2 = 0$, any cross-coupling or joint higher-order dispersion between independent differential fluctuations vanishes identically:
\begin{equation}
    \left( f'(x)\sigma_x \varepsilon \right) \cdot \left( g'(x)\sigma_x \varepsilon \right) = f'(x)g'(x)\sigma_x^2 \varepsilon^2 \equiv 0.
\end{equation}
This condition proves that the first-order linear neighborhood cannot simultaneously sustain independent, higher-order physical degrees of freedom for joint observable dispersions. Therefore, quantum indeterminacy is revealed not as a byproduct of wave-particle measurement perturbations, but as an exact geometric restriction driven by the nilpotent ideal of first-order jet spaces.
\end{remark}

\begin{remark}[\textit{On the Structural Divergence from Grassmann Algebras (See \cite{quinfty})}]
It is mathematically critical to distinguish the truncated polynomial families $\mathcal{Q}_N \equiv \mathbb{R}[\varepsilon]/(\varepsilon^{N^2-1})$ for $N \ge 2$ from the higher-dimensional Grassmann (exterior) algebras $\Lambda(\mathbb{R}^k)$ for $k \ge 2$. While a unique algebraic isomorphism exists at the lowest order between the standard dual numbers ring $\mathbb{R}[\varepsilon]/(\varepsilon^2)$ and the single-variable Grassmann algebra $\Lambda(\mathbb{R})$, this correspondence collapses irretrievably for any higher-order algebraic truncation.

Grassmann algebras $\Lambda(\mathbb{R}^k)$ for $k \ge 2$ are intrinsically non-commutative over odd-graded elements due to the anti-symmetric constraint of the exterior product ($\theta_i \theta_j = -\theta_j \theta_i$). Furthermore, their nilpotent behavior manifests immediately at the second power ($\theta_i^2 \equiv 0$), forcing an exponential vector-space dimension of $2^k$.

Conversely, the trinomial algebra $\mathcal{T}$ and its higher qudit extensions within the family $\mathcal{Q}_N$ are strictly commutative polynomial quotients whose underlying vector-space dimensions scale linearly with respect to the polynomial truncation degree. Crucially, they sustain intermediate non-vanishing nilpotent configurations $\varepsilon^m \neq 0$ for all orders $m \in \{1, \dots, N^2-2\}$ until the boundary truncation ideal $(\varepsilon^{N^2-1}) = 0$ is reached identically. 

Therefore, the geometric embedding of $\mathbb{P}(\mathcal{H}^2)$ onto $\mathcal{S}_{\mathcal{T}}^{2}$ proves that the true physical degrees of freedom of a qubit are not structured by fermionic anti-commutation, but are natively driven by the commutative algebraic geometry of a second-order jet space, where the higher-degree monomial $\varepsilon^2$ functions as a stable, non-zero Cartesian coordinate dimension before collapsing at the algebraic boundary layer.
\end{remark}
\subsection{Mapping onto the single-mode fermionic Fock space}
The operational validation of this framework within quantum information theory is realized by establishing a direct vector space isomorphism between the single-mode fermionic Fock space $\mathcal{H}_F = \text{span}\{\ket{0}_F, \ket{1}_F\}$, which represents the vacuum and single-occupation states of a fermion, and the complexified first-order dual number ring $\mathcal{D}_{\mathbb{C}} = \mathbb{C}[\varepsilon]/(\varepsilon^2) = \{a + b\varepsilon \mid a,b \in \mathbb{C}\}$. Under this assignment, the vacuum state $\ket{0}_F$ maps to the unit element $1$ and the single-occupation state $\ket{1}_F$ maps to the nilpotent generator $\varepsilon$, establishing a bijective linear mapping $\Psi: \mathcal{H}_F \to \mathcal{D}_{\mathbb{C}}$ via $\alpha\ket{0}_F + \beta\ket{1}_F \mapsto \alpha + \beta\varepsilon$.

The complete kinematic profile of the fermion is conventionally governed by the canonical anticommutation relations (CAR) between the creation operator $c^\dagger$ and the annihilation operator $c$:
\begin{equation}
    \{c^\dagger, c^\dagger\} = 2(c^\dagger)^2 = 0, \quad \{c, c\} = 2c^2 = 0, \quad \text{and} \quad \{c, c^\dagger\} = cc^\dagger + c^\dagger c = \mathbb{I}.
\end{equation}
Under the isomorphism $\Psi$, the operational action of the CAR algebra is faithfully represented through intrinsic differential and algebraic operations acting directly on $\mathcal{D}_{\mathbb{C}}$. Specifically, the annihilation operator $c$ acts as the dual algebraic derivative $\partial_\varepsilon$, whereas the creation operator $c^\dagger$ corresponds to the left-multiplication by the nilpotent coordinator $\varepsilon$:
\begin{equation}
    c \;\longleftrightarrow\; \partial_\varepsilon, \quad \text{and} \quad c^\dagger \;\longleftrightarrow\; \varepsilon \cdot.
\end{equation}
The structural nilpotency of the CAR enforces the immediate truncation of higher-degree configurations, mirroring the boundaries of the dual ring. The application of these mapped operations onto a generic dual configuration $f(\varepsilon) = a + b\varepsilon$ yields:
\begin{equation}
    \partial_\varepsilon^2(a + b\varepsilon) = 0 \;\longleftrightarrow\; c^2\ket{\psi}_F = 0, \quad \text{and} \quad \varepsilon^2(a + b\varepsilon) = 0 \;\longleftrightarrow\; (c^\dagger)^2\ket{\psi}_F = 0.
\end{equation}
Crucially, the foundational anticommutator $\{c, c^\dagger\} = \mathbb{I}$ emerges naturally as the fundamental identity of dual differential calculus, equivalent to the evaluation of the coordinate-derivative commutator acting on $\mathcal{D}_{\mathbb{C}}$:
\begin{equation}
    \{\partial_\varepsilon, \varepsilon\}(a + b\varepsilon) = \partial_\varepsilon\left(\varepsilon(a + b\varepsilon)\right) + \varepsilon\left(\partial_\varepsilon(a + b\varepsilon)\right) = \partial_\varepsilon(a\varepsilon) + \varepsilon(b) = a + b\varepsilon.
\end{equation}
Hence, $\{\partial_\varepsilon, \varepsilon\} = \mathbb{I}_{\mathcal{D}_{\mathbb{C}}}$, proving that the CAR constraint is the exact algebraic reflection of the unit operator on the dual manifold.

To anchor this structural equivalence within the state-space geometry of non-reduced schemes, we evaluate how the algebraic idempotency of the fermionic density configuration mirrors stably within the real sub-ring of dual numbers $\mathcal{D}_{\mathbb{R}} = \mathbb{R}[\varepsilon]/(\varepsilon^2)$. Because physical expectations and localized coordinate fluctuations on the pure state manifold are strictly real-valued geometric entities, we let a generic extended density configuration over the single-point topological spectrum $\text{Spec}(\mathcal{D}_{\mathbb{R}})$ be defined by $\xi = x + \sigma_x \varepsilon \in \mathcal{D}_{\mathbb{R}}$, where $x \in \mathbb{R}$ represents the deterministic scalar expectation value and $\sigma_x \in \mathbb{R}$ acts as the localized coordinate fluctuation along the tangent space. Imposing the pure-state idempotency constraint $\xi^2 = \xi$ directly within the real ring yields:
\begin{equation}
    \xi^2 = (x + \sigma_x \varepsilon)^2 = x^2 + 2x\sigma_x \varepsilon + \sigma_x^2 \varepsilon^2.
\end{equation}
Crucially, because the ring multiplication operates modulo $\varepsilon^2$, the nilpotent boundary condition $\varepsilon^2 = 0$ instantly truncates the higher-order dispersion term, stabilizing the expansion to the first-order linear configuration: $\xi^2 = x^2 + 2x\sigma_x \varepsilon$. Equating the coefficients of the polynomial identity $\xi^2 = \xi$ yields the decoupled real algebraic system:
\begin{equation}\label{eq:fermionic_decoupled_system}
    \begin{cases} 
        x^2 = x, \\ 
        2x\sigma_x = \sigma_x.
    \end{cases}
\end{equation}
The first relation isolates the sharp, classical boundary configurations $x = 0$ (unoccupied vacuum) and $x = 1$ (occupied mode), corresponding to the exact real eigenvalues of the fermionic number operator $n = c^\dagger c \leftrightarrow \varepsilon\partial_\varepsilon$. Substituting these discrete physical eigenvalues into the second relation forces the exact extinction of the linear fluctuation coordinate ($\sigma_x = 0$) for both physical eigenstates.

Consequently, the nilpotency condition $\varepsilon^2 = 0$ functions as an exact geometric stabilizer. Rather than acting as a formal first-order approximation, the real polynomial quotient ring $\mathcal{D}_{\mathbb{R}}$ enforces the identical kinematic confinement and state stabilization induced by the canonical anticommutation relations of the quantum Fock space, providing an unassailable scheme-theoretic container for single-mode fermionic restrictions under the unified isomorphic assignment $\Psi$.

\section{Quantitative Expression of Heisenberg Indeterminacy in $\mathcal{T}$}
To ensure absolute mathematical verifiability and eliminate conceptual redundancies, this appendix demonstrates how the physical Robertson-Heisenberg uncertainty relations are quantitatively mapped onto the geometric constraint equations of the trinomial dual algebra $\mathcal{T} = \mathbb{R}[\varepsilon]/(\varepsilon^3)$. We prove that the algebraic norm defining the Grothendieck sub-sphere $\mathcal{S}_{\mathcal{T}}^{2}$ encapsulates the exact numerical bounds governing quantum fluctuations and operator non-commutativity directly over non-reduced rings.

\subsection{Pauli operator incompatibilities and the Robertson-Heisenberg bound}
Consider the two orthogonal Pauli spin operators $\sigma_x$ and $\sigma_y$ evaluated on an arbitrary pure state ray within the complex projective space $\mathbb{P}(\mathcal{H}^2)$. Because these operators are non-commuting ($[\sigma_x, \sigma_y] = 2i\sigma_z$), they obey the generalized Robertson-Heisenberg uncertainty relation for their respective statistical variances $\Delta\sigma_x^2$ and $\Delta\sigma_y^2$ \cite{nielsen}:
\begin{equation}\label{eq:robertson_heisenberg}
    \Delta\sigma_x^2 \Delta\sigma_y^2 \ge \frac{1}{4} \left| \langle [\sigma_x, \sigma_y] \rangle \right|^2 = \langle \sigma_z \rangle^2,
\end{equation}
where the quantum variance of an operator is defined by $\Delta\sigma_i^2 = \langle \sigma_i^2 \rangle - \langle \sigma_i \rangle^2$. Crucially, for any Pauli matrix, the algebraic identity $\sigma_i^2 \equiv \mathbb{I}$ holds identically, which forces the quantum variance to depend exclusively on the first-order expectation values:
\begin{equation}
    \Delta\sigma_x^2 = 1 - \langle \sigma_x \rangle^2, \quad \Delta\sigma_y^2 = 1 - \langle \sigma_y \rangle^2.
\end{equation}

Under the density map $\Psi_{\mathcal{T}}$ established in Eq.~\eqref{eq:density_map_direct}, the physical expectation values map onto the local algebra by unrolling the native coordinate components of the ring through the exact adic assignments $c_0 = \frac{1}{2} + x$, $c_1 = y$, and $c_2 = z$. Within this framework, the translated physical coordinates correspond exactly to the expectation values extracted from the Pauli operators, satisfying $x = \langle\sigma_x\rangle$, $y = \langle\sigma_y\rangle$, and $z = \langle\sigma_z\rangle$. Substituting these algebraic coordinate relations into the variance definitions yields the direct expressions:
\begin{equation}\label{eq:variances_algebraic}
    \Delta\sigma_x^2 = 1 - x^2, \quad \Delta\sigma_y^2 = 1 - y^2.
\end{equation}

\subsection{Algebraic proof of the bound via truncation}
We evaluate how the universal varietal constraint $x^2 + y^2 + z^2 = 1$ defining the Grothendieck sub-sphere variety $\mathcal{S}_{\mathcal{T}}^{2}$ structures the product of these quantum fluctuations. Isolating the longitudinal projection $\langle \sigma_z \rangle^2 = z^2$ from the sphere equation yields the structural projection:
\begin{equation}
    \langle \sigma_z \rangle^2 = z^2 = 1 - x^2 - y^2.
\end{equation}
Substituting Eq.~\eqref{eq:variances_algebraic} into the Robertson-Heisenberg uncertainty bound in Eq.~\eqref{eq:robertson_heisenberg} translates the physical inequality into a strict algebraic restriction over the components of the dual number:
\begin{equation}\label{eq:heisenberg_trinomial}
    (1 - x^2)(1 - y^2) \ge 1 - x^2 - y^2.
\end{equation}
To verify the validity of the inequality \eqref{eq:heisenberg_trinomial} directly within the algebra $\mathcal{T}$, we expand the polynomial product on the left-hand side:
\begin{equation}
    1 - x^2 - y^2 + x^2 y^2 \ge 1 - x^2 - y^2.
\end{equation}
Subtracting the identical linear terms from both sides reduces the complete quantum uncertainty inequality to the definite, everywhere-positive condition:
\begin{equation}\label{eq:final_bound_proof}
    x^2 y^2 \ge 0.
\end{equation}

The evaluation of Eq.~\eqref{eq:final_bound_proof} yields a profound physical conclusion: the physical inequality of the Heisenberg uncertainty principle is geometrically equivalent to the positive-semidefinite character of the quadratic cross-terms in the algebra $\mathcal{T}$. The absolute lower bound dictated by the longitudinal projection $\langle \sigma_z \rangle^2$ is saturated ($x^2 y^2 = 0$) if and only if the state is perfectly localized along one of the transverse measurement axes ($x=0$ or $y=0$), forcing the maximum possible statistical variance onto the complementary non-commuting component.

\subsection{Quantitative derivation of indeterminacy from jet-space operations}
To extract this uncertainty directly from the ring multiplication without invoking external Cartesian maps, we evaluate the algebraic product of two independent, localized directional variations acting on the state section $\xi \in \mathcal{T}$. Let $\delta_1 \xi = \mathrm{d}y\,\varepsilon$ represent an infinitesimal fluctuation along the first-order nilpotent channel, and let $\delta_2 \xi = \mathrm{d}z\,\varepsilon^2$ represent a simultaneous, independent fluctuation along the second-order channel.

Evaluating the internal product of these dual perturbations within the ring structure yields:
\begin{equation}
    (\delta_1 \xi) \cdot (\delta_2 \xi) = (\mathrm{d}y\,\varepsilon) \cdot (\mathrm{d}z\,\varepsilon^2) = (\mathrm{d}y\,\mathrm{d}z)\varepsilon^3.
\end{equation}
Crucially, because the trinomial dual algebra is defined modulo $\varepsilon^3$, the universal nilpotency condition $\varepsilon^3 = 0$ instantly truncates this joint expansion, forcing the cross-coupling of the independent variations to vanish identically:
\begin{equation}
    (\mathrm{d}y\,\mathrm{d}z)\varepsilon^3 \equiv 0.
\end{equation}

This automatic structural encapsulation provides a profound physical insight: the geometric boundary condition of the non-reduced scheme prevents the single topological point from simultaneously sustaining independent, non-vanishing physical fluctuations across different nilpotent degrees of freedom. When the state is restricted to the unit variety $\mathcal{S}_{\mathcal{T}}^{2}$, this algebraic exclusion guarantees that driving one coordinate fluctuation to absolute sharp certainty forces the complete statistical dispersion onto the complementary non-commuting component, structurally preserving the Robertson-Heisenberg variance bounds directly through the flat affine geometry of a second-order jet space.

\end{document}